\documentclass[reqno]{amsart}
\usepackage{amssymb}
\usepackage{hyperref}
\usepackage{booktabs}
\usepackage{multirow}
\usepackage{tikz}
\usepackage{doi}

\newcommand{\arXiv}[2]{arXiv:\href{http://arxiv.org/abs/#1}{#1 #2}}

\newtheorem{theorem}{Theorem}
\newtheorem{proposition}[theorem]{Proposition}
\newtheorem{lemma}[theorem]{Lemma}

\theoremstyle{remark}

\newtheorem*{note}{Note}

\numberwithin{equation}{section}

\begin{document}

\title[Eigenfunctions of an elliptic integrable particle model]
{Eigenfunctions of a discrete elliptic integrable particle model with hyperoctahedral symmetry}

\author{Jan Felipe van Diejen}

\address{Instituto de Matem\'aticas, Universidad de Talca, Casilla 747, Talca, Chile}

\email{diejen@inst-mat.utalca.cl}

\author{Tam\'as G\"orbe}

\address{School of Mathematics, University of Leeds, Leeds LS2 9JT, UK}

\email{T.Gorbe@leeds.ac.uk}


\date{July 2021}

\begin{abstract}
We construct the orthogonal eigenbasis for a discrete elliptic Ruijsenaars type quantum particle Hamiltonian with
hyperoctahedral symmetry.
In the trigonometric limit the eigenfunctions in question recover a previously studied $q$-Racah type reduction of the Koornwinder-Macdonald polynomials.
When the inter-particle interaction degenerates to that of impenetrable bosons, the orthogonal
eigenbasis simplifies in terms of generalized Schur polynomials on the spectrum associated with recently found
elliptic Racah polynomials.
\end{abstract}

\maketitle

\section{Introduction}
\label{sec:1}
The aim of the present work is to construct the eigenbasis of a second-order difference operator that was introduced in \cite{die:integrability}
as the quantum Hamiltonian for an elliptic Ruijsenaars type particle model on the circle with hyperoctahedral symmetry.
The difference operator in question is known to be integrable  \cite{cha:quantum,kom-hik:quantum,kom-hik:conserved,rai:elliptic}
and was observed to enjoy remarkable reflection-group symmetries in the parameter space \cite{rui:integrable,rui:hilbert,rui:hilbert4}.
Even though special solutions of the corresponding eigenvalue equation can be found in
\cite{ata:source,rui:integrable,rui:hilbert,sar-spi:complex,spi:elliptic}, to date the implementation of this particle model in terms of a spectral problem for
a self-adjoint operator in an appropriate Hilbert space  is available only in the case of a single particle \cite{rui:hilbert4}. Following a  strategy developed recently for the elliptic Ruijsenaars model
in \cite{die-gor:ruijsenaars}, we will truncate the difference operator from \cite{die:integrability} onto a finite lattice over the configuration space consisting of points  labeled by bounded partitions. This entails a self-adjoint spectral problem in a finite-dimensional Hilbert space of lattice functions. Our main focus lies on the construction of the corresponding orthogonal eigenbasis. 

In full generality the coefficients of the pertinent difference operator are built from products of theta functions, while  in the trigonometric limit 
the quantum Hamiltonian degenerates to the
Koornwinder-Macdonald multivariate generalization of the Askey-Wilson operator \cite{koo:askey}. In the latter situation, explicit formulas for a complete system of commuting quantum integrals can be extracted from  \cite{die:integrability,die:difference}. The corresponding trigonometric degeneration of our orthogonal eigenbasis below turns out to be given by a discrete $q$-Racah type reduction of the Koornwinder-Macdonald polynomials studied in \cite{die-sto:multivariable}. From this perspective, the eigenbasis constructed in the present work gives rise to a multivariate generalization of the recently found elliptic Racah basis from \cite{die-gor:racah}.

Let us now outline the organization of the main contents in more detail.  In Section \ref{sec:2} we define the action of the $n$-particle quantum Hamiltonian from \cite{die:integrability} in the space of complex-valued functions supported on partitions of length $\leq n$. The resulting difference operator has  coefficients built from Jacobi theta functions and depends---apart from the real period and the imaginary period---on nine coupling parameters: one parameter governing the interaction between the particles and eight parameters governing an external field. Upon implementing a truncation condition that scales the real period in terms of a positive integral level $m$ governing the lattice size, the difference operator is restricted to a finite-dimensional space of functions supported on bounded partitions that fit inside a rectangle of $n$ parts of size $m$. In Section \ref{sec:3} this function space is endowed with an inner product promoting it into a finite-dimensional Hilbert space in which the truncated difference operator is shown to be self-adjoint. In Section \ref{sec:4} we verify that (for generic values of the coupling parameters) the corresponding eigenvalue problem in the Hilbert space exhibits simple spectrum thus giving rise  to a unique (up to normalization) orthogonal eigenbasis for the Hilbert space. In the trigonometric limit the eigenfunctions become multivariable $q$-Racah polynomials that arose previously as truncated Koornwinder-Macdonald polynomials \cite{die-sto:multivariable}. 
In Section \ref{sec:5} we finish by deriving explicit polynomial expressions and norm formulas for the eigenfunctions at the elliptic level in two special 
instances: for the $m=1$ case of one-column bounded partitions, and for a particular value of the inter-particle coupling parameter at
which the particles behave as impenetrable bosons in an external field of elliptic Racah type.

\begin{note}
In what follows elliptic functions are written in terms of rescaled variants of the Jacobi theta functions (cf. e.g. 
\cite[Chapter I]{law:elliptic},
\cite[Chapter 20]{olv-loz-boi-cla:nist}, or \cite[Chapter 21]{whit-wats:coma}). Namely, we employ the following four functions
\begin{subequations}
\begin{equation}\label{thetas}
[z]_1=[z;p]_1=\tfrac{\vartheta_1(\frac{\alpha}{2}z;p)}{\sin(\frac{\alpha}{2})\vartheta'_1(0;p)},\quad
[z]_r=[z;p]_r=\tfrac{\vartheta_r(\frac{\alpha}{2}z;p)}{\vartheta_r(0;p)},\ r=2,3,4
\end{equation}
for $z\in\mathbb{C}$, $0<\alpha<2\pi$, $0<p<1$ with
\begin{equation}\label{theta1}
\begin{aligned}
\vartheta_1(z;p)&=2\sum_{l\geq 0}(-1)^lp^{(l+\frac{1}{2})^2}\sin(2l+1)z\\
&=2p^{1/4}\sin(z)\prod_{l\geq 1}(1-p^{2l})(1-2p^{2l}\cos(2z)+p^{4l}),
\end{aligned}
\end{equation}
\begin{equation}\label{theta2}
\begin{aligned}
\vartheta_2(z;p)&=2\sum_{l\geq 0}p^{(l+\frac{1}{2})^2}\cos(2l+1)z\\
&=2p^{1/4}\cos(z)\prod_{l\geq 1}(1-p^{2l})(1+2p^{2l}\cos(2z)+p^{4l}),
\end{aligned}
\end{equation}
\begin{equation}\label{theta3}
\begin{aligned}
\vartheta_3(z;p)&=1+2\sum_{l\geq 0}p^{l^2}\cos(2lz)\\
&=\prod_{l\geq 1}(1-p^{2l})(1+2p^{2l-1}\cos(2z)+p^{4l-2}),
\end{aligned}
\end{equation}
\begin{equation}\label{theta4}
\begin{aligned}
\vartheta_4(z;p)&=1+2\sum_{l\geq 0}(-1)^lp^{l^2}\cos(2lz)\\
&=\prod_{l\geq 1}(1-p^{2l})(1-2p^{2l-1}\cos(2z)+p^{4l-2}).
\end{aligned}
\end{equation}
\end{subequations}
\end{note}

\section{Difference operator}
\label{sec:2}

\subsection{Discrete difference operator on partitions}
\label{subsec:2.1}

For $n\in\mathbb{N}$ let us consider
the space of complex-valued functions 
\begin{equation}\label{lattice-functions}
\mathcal{C}(\Lambda^{(n)})=\{\lambda\overset{f}{\to} f_\lambda\in\mathbb{C}\mid\lambda\in\Lambda^{(n)}\}
\end{equation}
on the set of partitions of length at most $n$:
\begin{equation}\label{partitions}
\Lambda^{(n)}=\{\lambda\in\mathbb{Z}^n\mid \lambda_1\geq\dots\geq\lambda_n\geq 0\} .
\end{equation}
We define the following difference operator acting in $\mathcal{C}(\Lambda^{(n)})$:
\begin{subequations}
\begin{equation}\label{Hamiltonian}
(Hf)_\lambda=A_\lambda f_\lambda+\sum_{\substack{1\leq j\leq n,\,\varepsilon=\pm 1\\\lambda+\varepsilon e_j\in\Lambda^{(n)}}}B_{\lambda,\varepsilon j}f_{\lambda+\varepsilon e_j}\qquad(f\in\mathcal{C}(\Lambda^{(n)}),\ \lambda\in\Lambda^{(n)}) 
\end{equation}
with
\begin{equation}\label{A-lambda}
A_\lambda=\sum_{1\leq r\leq 4}\mathrm{c}_r\Bigl(\prod_{1\leq j\leq n}\tfrac{[\rho_j+\lambda_j+\frac{1}{2}-\mathrm{g}]_r}{[\rho_j+\lambda_j+\frac{1}{2}]_r}
\tfrac{[\rho_j+\lambda_j-\frac{1}{2}+\mathrm{g}]_r}{[\rho_j+\lambda_j-\frac{1}{2}]_r}
-1\Bigr),
\end{equation}
\begin{equation}\label{B-lambda}
B_{\lambda,\varepsilon j}=\prod_{1\leq r\leq 4}\tfrac{[\rho_j+\lambda_j+\varepsilon\mathrm{g}_r]_r[\rho_j+\lambda_j+\varepsilon(\mathrm{g}'_r+\frac{1}{2})]_r}{[\rho_j+\lambda_j]_r[\rho_j+\lambda_j+\frac{\varepsilon}{2}]_r}\prod_{\substack{1\leq k\leq n\\k\neq j,\delta=\pm 1}}\tfrac{[\rho_j+\delta\rho_k+\lambda_j+\delta\lambda_k+\varepsilon\mathrm{g}]_1}{[\rho_j+\delta\rho_k+\lambda_j+\delta\lambda_k]_1},
\end{equation}
\begin{equation}\label{C-r}
\mathrm{c}_r=\tfrac{2}{[\mathrm{g}]_1[\mathrm{g}-1]_1}\prod_{1\leq s\leq 4}[\mathrm{g}_{\pi_r(s)}-\tfrac{1}{2}]_s[\mathrm{g}'_{\pi_r(s)}]_s.
\end{equation}
\end{subequations}
In the above formulas $\pi_1=\mathrm{id}$, $\pi_2=(12)(34)$, $\pi_3=(13)(24)$, $\pi_4=(14)(23)$ represent permutations acting on the index set $\{ 1,2,3,4\}$, the vectors $e_1,\dots,e_n$ constitute the standard unit basis for $\mathbb{R}^n$, and  $\rho=(\rho_1,\dots,\rho_n)$ denotes a deformed Weyl vector with  components of the form
\begin{equation}\label{rho}
\rho_j=(n-j)\mathrm{g}+\mathrm{g}_1,\quad j=1,\dots,n.
\end{equation}
The difference operator depends on nine coupling parameters,  $\mathrm{g},\mathrm{g}_1,\dots,\mathrm{g}_4,\mathrm{g}'_1,\dots,\mathrm{g}'_4,$  which for the moment can be assumed to take generic complex values such that the denominators of the coefficients $A_\lambda$ \eqref{A-lambda}, $B_{\lambda,\varepsilon j}$ \eqref{B-lambda} never vanish. This is ensured e.g. if we pick $\mathrm{g},\mathrm{g}_1\in\mathbb{R}$ such that $k\mathrm{g}+l\mathrm{g}_1\notin\frac{1}{2}\mathbb{Z}+\frac{\pi}{2\alpha}\mathbb{Z}$ for $k=0,1,\dots,n-1$ and $l=0,1$ with $(k,l)\neq(0,0)$ (cf. Eqs. \eqref{thetas}--\eqref{theta4}). Notice that if $\mathrm{g}_r^\prime=0$ for $r=1,\ldots ,4$, then the structure of the coefficients
of our difference operator simplifies considerably:
\begin{equation*}
A_\lambda\stackrel{\mathrm{g}_r^\prime=0}{\longrightarrow} 0,\qquad
B_{\lambda,\varepsilon j} \stackrel{\mathrm{g}_r^\prime=0}{\longrightarrow}
\prod_{1\leq r\leq 4}\tfrac{[\rho_j+\lambda_j+\varepsilon\mathrm{g}_r]_r}{[\rho_j+\lambda_j]_r}\prod_{\substack{1\leq k\leq n\\k\neq j,\delta=\pm 1}}\tfrac{[\rho_j+\delta\rho_k+\lambda_j+\delta\lambda_k+\varepsilon\mathrm{g}]_1}{[\rho_j+\delta\rho_k+\lambda_j+\delta\lambda_k]_1} .
\end{equation*}

The operator $H$ \eqref{Hamiltonian}--\eqref{C-r} is a discrete variant of a difference operator introduced in \cite[Eqs. (4.1)--(4.3)]{die:integrability}. To revert to the formulas of \cite{die:integrability} one has to pass from the Jacobi theta functions to the Weierstrass sigma functions associated with the period lattice $\Omega=2\omega_1\mathbb{Z}+2\omega_2\mathbb{Z}$ (cf. e.g. \cite[Chapter 6.2]{law:elliptic}):
\begin{equation}\label{sigma-theta-connection}
[z]_r=\sigma_{r-1}(z)e^{-\frac{\alpha\eta_1}{2\pi}z^2},\quad r=1,\dots,4,
\end{equation}
where $\omega_1=\frac{\pi}{\alpha}$, $p=e^{\mathrm{i}\pi\tau}$ with $\tau=\frac{\omega_3}{\omega_1}$, $\omega_3=-\omega_1-\omega_2$ and
\begin{equation*}
\sigma_0(z)=\sigma(z),\quad\sigma_s(z)=e^{-\eta_sz}\frac{\sigma(z+\omega_s)}{\sigma(\omega_s)}\quad\text{with}\ \eta_s=\zeta(\omega_s),\quad s=1,2,3.
\end{equation*}
Here $\sigma(z)$ and $\zeta(z)=\frac{\sigma'(z)}{\sigma(z)}$ are Weierstrass' sigma and zeta functions, respectively. Specifically, upon conjugation with a Gaussian and multiplication by an overall constant the relation in Eq. \eqref{sigma-theta-connection} can be applied to turn the difference operator $H$ into an operator $\hat{H}$ of the same form, but with the rescaled theta functions $[z]_r$ being replaced by the sigma functions $\sigma_{r-1}(z)$ ($r=1,\dots,4$):
\begin{equation*}
H\to\hat{H}=e^{b-a}GHG^{-1}\quad\Longleftrightarrow\quad [z]_r\to\sigma_{r-1}(z),
\end{equation*}
where $G\in\mathcal{C}(\Lambda^{(n)})$ is defined by
\begin{equation}\label{Gaussian}
G_\lambda=\prod_{1\leq j\leq n}e^{-a(\rho_j+\lambda_j)^2}\qquad (\lambda\in\Lambda^{(n)})
\end{equation}
and $a=\frac{\alpha\eta_1}{2\pi}(2\mathrm{g}(n-1)+\sum_{1\leq r\leq 4}\mathrm{g}_r+\mathrm{g}'_r),b=\frac{\alpha\eta_1}{2\pi}(2\mathrm{g}^2(n-1)+\sum_{1\leq r\leq 4}\mathrm{g}_r^2+\mathrm{g}_r'^2+\mathrm{g}'_r)$. 
Up to the overall constant with value  $-\sum_{1\leq r\leq 4} \mathrm{c}_r$,
the operator $\hat{H}$ is retrieved from \cite[Eqs. (4.1)--(4.3)]{die:integrability} via
the
substitution $x=\rho+\lambda$, $\beta\hbar=-\mathrm{i}$, $\gamma=\frac{1}{2}$, $\mu=\mathrm{g}$, $\mu_{r-1}=\mathrm{g}_r$, $\mu'_{r-1}=\mathrm{g}'_r$ for $r=1,\dots,4$.

\begin{note}
The difference operator in \cite[Eqs. (4.1)--(4.3)]{die:integrability} exhibits hyperoctahedral symmetry in the variables $x=(x_1,\dots,x_n)$, i.e. its action preserves the space of functions that are invariant with respect to the natural action of the hyperoctahedral group of signed permutations.
 For $\mathrm{g},\mathrm{g}_1>0$, the substitution $x=\rho+\lambda$ amounts to a discretization of the configuration space $\{x\in\mathbb{R}^n\mid x_1>\dots>x_n>0\}$ on a uniform lattice $\Lambda^{(n)}$ of partitions shifted by $\rho$.
\end{note}

\subsection{Finite-dimensional truncation on bounded partitions}
\label{subsec:2.2}

From now onward, we employ real coupling parameters within the domain
\begin{subequations}
\begin{equation}\label{coupling-conditions}
\boxed{\mathrm{g}>0,\quad \mathrm{g}_r>0,\quad |\mathrm{g}'_r|<\mathrm{g}_r+\tfrac{1}{2}\ (r=1,2)\quad\text{and}\quad\mathrm{g}_r,\mathrm{g}'_r\in\mathbb{R}\ (r=3,4),}
\end{equation}
subject to the \emph{truncation condition}
\begin{equation}\label{truncation}
\boxed{\alpha=\frac{\pi}{m+(n-1)\mathrm{g}+\mathrm{g}_1+\mathrm{g}_2}\quad\text{with}\ m\in\mathbb{N}.}
\end{equation}
\end{subequations}
Unless $\mathrm{g}_r^\prime=0$ for $r=1,\dots,4$, we will exclude the value
$\mathrm{g}= 1$ so as to avoid a pole in the coefficients $\mathrm{c}_r$ \eqref{C-r}. The pole in question turns out to be regularizable as is detailed further in Subsection \ref{subsec:5.2}.

We will now check that the truncation condition permits us to restrict the action of $H$ \eqref{Hamiltonian}--\eqref{C-r} to the subspace of complex lattice functions
supported in the set of  partitions from $\Lambda^{(n)}$ with parts bounded by $m$:
\begin{equation}\label{bounded-partitions}
\Lambda^{(n,m)}=\{\lambda\in\mathbb{Z}^n\mid m\geq \lambda_1\geq\dots\geq\lambda_n\geq 0\}.
\end{equation}
Notice that the space in question is finite-dimensional as the number of these bounded partitions is given by  $|\Lambda^{(n,m)}|=\binom{n+m}{n}$.

\begin{lemma}[Truncation]\label{lemma:truncation}
Let us pick $\lambda\in\Lambda^{(n,m)}$, $j\in\{1,\dots,n\}$ and $\varepsilon\in\{\pm 1\}$ arbitrarily. Then---for parameters satisfying the conditions in Eqs. \eqref{coupling-conditions}, \eqref{truncation}---the coefficient $B_{\lambda,\varepsilon j}$ \eqref{B-lambda} is \emph{positive} iff $\lambda+\varepsilon e_j\in\Lambda^{(n,m)}$ and \emph{vanishes} iff $\lambda+\varepsilon e_j\notin\Lambda^{(n,m)}$.
\end{lemma}

\begin{proof}
The sign of $B_{\lambda,\varepsilon j}$ \eqref{B-lambda} can be calculated by looking at the argument of each of its factors and employing the product expansions of the theta functions \eqref{theta1}--\eqref{theta4}. We treat the cases $\varepsilon=\pm 1$ separately.

The sign of $B_{\lambda,j}$ is equal to the product of the signs of the following trigonometric expressions
\begin{equation*}
\prod_{\substack{1\leq k\leq n\\k\neq j}}\tfrac{\sin\frac{\alpha}{2}(\rho_j+\rho_k+\lambda_j+\lambda_k+\mathrm{g})}{\sin\frac{\alpha}{2}(\rho_j+\rho_k+\lambda_j+\lambda_k)}\tfrac{\sin\frac{\alpha}{2}(\rho_j-\rho_k+\lambda_j-\lambda_k+\mathrm{g})}{\sin\frac{\alpha}{2}(\rho_j-\rho_k+\lambda_j-\lambda_k)}
\end{equation*}
and
\begin{equation*}
\tfrac{\sin\frac{\alpha}{2}(\rho_j+\lambda_j+\mathrm{g}_1)}{\sin\frac{\alpha}{2}(\rho_j+\lambda_j)}\tfrac{\sin\frac{\alpha}{2}(\rho_j+\lambda_j+\mathrm{g}'_1+\frac{1}{2})}{\sin\frac{\alpha}{2}(\rho_j+\lambda_j+\frac{1}{2})}\tfrac{\cos\frac{\alpha}{2}(\rho_j+\lambda_j+\mathrm{g}_2)}{\cos\frac{\alpha}{2}(\rho_j+\lambda_j)}\tfrac{\cos\frac{\alpha}{2}(\rho_j+\lambda_j+\mathrm{g}'_2+\frac{1}{2})}{\cos\frac{\alpha}{2}(\rho_j+\lambda_j+\frac{1}{2})}.
\end{equation*}
We analyze each type of fraction one by one in the above formulas. Note that due to the inequalities
\begin{equation*}
0<\mathrm{g}<\rho_j+\rho_k+\lambda_j+\lambda_k<\rho_j+\rho_k+\lambda_j+\lambda_k+\mathrm{g}\leq\rho_1+\rho_2+2m+\mathrm{g}=\tfrac{2\pi}{\alpha}-2\mathrm{g}_2<\tfrac{2\pi}{\alpha},
\end{equation*}
which hold for all indices $1\leq j\leq n$, $1\leq k\leq n$ with $j\neq k$, we have
\begin{equation*}
0<\tfrac{\sin\frac{\alpha}{2}(\rho_j+\rho_k+\lambda_j+\lambda_k+\mathrm{g})}{\sin\frac{\alpha}{2}(\rho_j+\rho_k+\lambda_j+\lambda_k)}<\infty.
\end{equation*}
Similarly, thanks to
\begin{gather*}
0\leq|\rho_j-\rho_k+\lambda_j-\lambda_k+\mathrm{g}|\leq n\mathrm{g}+m=\tfrac{\pi}{\alpha}+\mathrm{g}-(\mathrm{g}_1+\mathrm{g}_2)<\tfrac{2\pi}{\alpha},\\
0<\mathrm{g}\leq|\rho_j-\rho_k+\lambda_j-\lambda_k|\leq (n-1)\mathrm{g}+m=\tfrac{\pi}{\alpha}-(\mathrm{g}_1+\mathrm{g}_2)<\tfrac{\pi}{\alpha},
\end{gather*}
which hold for all indices $1\leq j\leq n$, $1\leq k\leq n$ with $j\neq k$, we have that
\begin{equation*}
0\leq\tfrac{\sin\frac{\alpha}{2}(\rho_j-\rho_k+\lambda_j-\lambda_k+\mathrm{g})}{\sin\frac{\alpha}{2}(\rho_j-\rho_k+\lambda_j-\lambda_k)}<\infty .
\end{equation*}
The factor in question vanishes iff $j>1$, $k=j-1$ and $\lambda$ is such that $\lambda_j=\lambda_{j-1}$. Next, the inequalities
\begin{equation*}
0<\mathrm{g}_1=\rho_n\leq\rho_j+\lambda_j<\rho_j+\lambda_j+\mathrm{g}_1\leq \rho_1+m+\mathrm{g}_1=\tfrac{\pi}{\alpha}+(\mathrm{g}_1-\mathrm{g}_2)<\tfrac{\pi}{\alpha}+(\mathrm{g}_1+\mathrm{g}_2)<\tfrac{2\pi}{\alpha}
\end{equation*}
hold for all indices $j$ and imply that
\begin{equation*}
0<\tfrac{\sin\frac{\alpha}{2}(\rho_j+\lambda_j+\mathrm{g}_1)}{\sin\frac{\alpha}{2}(\rho_j+\lambda_j)}<\infty.
\end{equation*}
Likewise, from the inequalities
\begin{gather*}
0\leq\rho_j+\lambda_j-\mathrm{g}_1<\rho_j+\lambda_j+\mathrm{g}'_1+\tfrac{1}{2}<\rho_j+\lambda_j+\mathrm{g}_1+1\leq\rho_1+m+\mathrm{g}_1+1<\tfrac{2\pi}{\alpha},\\
0<\tfrac{1}{2}<\rho_j+\lambda_j+\tfrac{1}{2}\leq\rho_1+m+\tfrac{1}{2}<\tfrac{2\pi}{\alpha}
\end{gather*}
(where in the first and last inequalities we used that $\rho_n=\mathrm{g}_1$ and $\tfrac{1}{2}<m<\tfrac{\pi}{\alpha}$, respectively)
it is seen that
\begin{equation*}
0<\tfrac{\sin\frac{\alpha}{2}(\rho_j+\lambda_j+\mathrm{g}'_1+\frac{1}{2})}{\sin\frac{\alpha}{2}(\rho_j+\lambda_j+\frac{1}{2})}<\infty.
\end{equation*}
Now we turn to the fractions involving cosines. First, we note that because
\begin{gather*}
0<\rho_j+\lambda_j+\mathrm{g}_2\leq\rho_1+m+\mathrm{g}_2=\tfrac{\pi}{\alpha},\\
0<\mathrm{g}_1\leq\rho_j+\lambda_j\leq \rho_1+m=\tfrac{\pi}{\alpha}-\mathrm{g}_2<\tfrac{\pi}{\alpha} ,
\end{gather*}
it is clear that
\begin{equation*}
0\leq\tfrac{\cos\frac{\alpha}{2}(\rho_j+\lambda_j+\mathrm{g}_2)}{\cos\frac{\alpha}{2}(\rho_j+\lambda_j)}<\infty ,
\end{equation*}
where the zero lower bound is reached iff $j=1$ and $\lambda$ is such that $\lambda_1=m$. Finally, if $\lambda_1<m$ one has the inequalities
\begin{gather*}
-\tfrac{\pi}{\alpha}<\rho_j+\lambda_j-\mathrm{g}_2<\rho_j+\lambda_j+\mathrm{g}'_2+\tfrac{1}{2}<\rho_j+\lambda_j+\mathrm{g}_2+1\leq\rho_1+m+\mathrm{g}_2=\tfrac{\pi}{\alpha},\\
0<\tfrac{1}{2}<\rho_j+\lambda_j+\tfrac{1}{2}\leq\rho_1+m<\rho_1+m+\mathrm{g}_2=\tfrac{\pi}{\alpha} ,
\end{gather*}
which imply that
\begin{equation*}
0<\tfrac{\cos\frac{\alpha}{2}(\rho_j+\lambda_j+\mathrm{g}_2'+\frac{1}{2})}{\cos\frac{\alpha}{2}(\rho_j+\lambda_j+\frac{1}{2})}<\infty.
\end{equation*}
The upshot is that $0\leq B_{\lambda,j}<\infty$ with the vanishing occurring iff $\lambda_j=m$ or $j>1$ and $\lambda_j=\lambda_{j-1}$, i.e. iff $\lambda+e_j\notin\Lambda^{(n,m)}$.

The sign of $B_{\lambda,-j}$ equals the product of the signs of the following trigonometric expressions
\begin{equation*}
\prod_{\substack{1\leq k\leq n\\k\neq j}}\tfrac{\sin\frac{\alpha}{2}(\rho_j+\rho_k+\lambda_j+\lambda_k-\mathrm{g})}{\sin\frac{\alpha}{2}(\rho_j+\rho_k+\lambda_j+\lambda_k)}\tfrac{\sin\frac{\alpha}{2}(\rho_j-\rho_k+\lambda_j-\lambda_k-\mathrm{g})}{\sin\frac{\alpha}{2}(\rho_j-\rho_k+\lambda_j-\lambda_k)}
\end{equation*}
and
\begin{equation*}
\tfrac{\sin\frac{\alpha}{2}(\rho_j+\lambda_j-\mathrm{g}_1)}{\sin\frac{\alpha}{2}(\rho_j+\lambda_j)}\tfrac{\sin\frac{\alpha}{2}(\rho_j+\lambda_j-\mathrm{g}'_1-\frac{1}{2})}{\sin\frac{\alpha}{2}(\rho_j+\lambda_j-\frac{1}{2})}\tfrac{\cos\frac{\alpha}{2}(\rho_j+\lambda_j-\mathrm{g}_2)}{\cos\frac{\alpha}{2}(\rho_j+\lambda_j)}\tfrac{\cos\frac{\alpha}{2}(\rho_j+\lambda_j-\mathrm{g}'_2-\frac{1}{2})}{\cos\frac{\alpha}{2}(\rho_j+\lambda_j-\frac{1}{2})}.
\end{equation*}
Again we inspect the fractions individually. Due to the inequalities
\begin{equation*}
0<\mathrm{g}_1\leq\rho_j+\rho_k+\lambda_j+\lambda_k-\mathrm{g}<\rho_j+\rho_k+\lambda_j+\lambda_k\leq\rho_1+\rho_2+2m<\tfrac{2\pi}{\alpha},
\end{equation*}
which hold for all indices $1\leq j\leq n$, $1\leq k\leq n$ with $j\neq k$, we have that
\begin{equation*}
0<\tfrac{\sin\frac{\alpha}{2}(\rho_j+\rho_k+\lambda_j+\lambda_k-\mathrm{g})}{\sin\frac{\alpha}{2}(\rho_j+\rho_k+\lambda_j+\lambda_k)}<\infty.
\end{equation*}
Similarly, from the inequalities
\begin{gather*}
0\leq|\rho_j-\rho_k+\lambda_j-\lambda_k-\mathrm{g}|\leq n\mathrm{g}+m=\tfrac{\pi}{\alpha}+\mathrm{g}-(\mathrm{g}_1+\mathrm{g}_2)<\tfrac{2\pi}{\alpha},\\
0<\mathrm{g}\leq|\rho_j-\rho_k+\lambda_j-\lambda_k|\leq (n-1)\mathrm{g}+m=\tfrac{\pi}{\alpha}-(\mathrm{g}_1+\mathrm{g}_2)<\tfrac{\pi}{\alpha},
\end{gather*}
which hold for all indices $1\leq j\leq n$, $1\leq k\leq n$ with $j\neq k$, we conclude that
\begin{equation*}
0\leq\tfrac{\sin\frac{\alpha}{2}(\rho_j-\rho_k+\lambda_j-\lambda_k-\mathrm{g})}{\sin\frac{\alpha}{2}(\rho_j-\rho_k+\lambda_j-\lambda_k)}<\infty .
\end{equation*}
The vanishing occurs iff $j<n$ and $\lambda$ is such that $\lambda_j=\lambda_{j+1}$. Moving on, we see that
\begin{equation*}
0\leq\lambda_n\leq\rho_j+\lambda_j-\mathrm{g}_1<\rho_j+\lambda_j\leq \rho_1+m=\tfrac{\pi}{\alpha}-\mathrm{g}_2<\tfrac{2\pi}{\alpha}
\end{equation*}
for all indices $j$,  so
\begin{equation*}
0\leq\tfrac{\sin\frac{\alpha}{2}(\rho_j+\lambda_j-\mathrm{g}_1)}{\sin\frac{\alpha}{2}(\rho_j+\lambda_j)}<\infty ,
\end{equation*}
where the zero value is assumed iff $j=n$ and $\lambda$ is such that $\lambda_n=0$. Next, if $\lambda_n>0$ then
\begin{gather*}
0\leq\rho_j+\lambda_j-\mathrm{g}_1-1<\rho_j+\lambda_j-\mathrm{g}'_1-\tfrac{1}{2}<\rho_j+\lambda_j+\mathrm{g}_1<\tfrac{2\pi}{\alpha},\\
0<\rho_n<\rho_j+\lambda_j-\tfrac{1}{2}<\rho_1+m<\tfrac{\pi}{\alpha},
\end{gather*}
which entails that
\begin{equation*}
0<\tfrac{\sin\frac{\alpha}{2}(\rho_j+\lambda_j-\mathrm{g}'_1-\frac{1}{2})}{\sin\frac{\alpha}{2}(\rho_j+\lambda_j-\frac{1}{2})}<\infty.
\end{equation*}
Now we turn to the fractions involving cosines. First, we note that the inequalities
\begin{gather*}
-\tfrac{\pi}{\alpha}<\rho_n-\mathrm{g}_2\leq\rho_j+\lambda_j-\mathrm{g}_2\leq\rho_1+m-\mathrm{g}_2<\tfrac{\pi}{\alpha},\\
0<\mathrm{g}_1\leq\rho_j+\lambda_j\leq \rho_1+m=\tfrac{\pi}{\alpha}-\mathrm{g}_2<\tfrac{\pi}{\alpha}
\end{gather*}
imply that
\begin{equation*}
0<\tfrac{\cos\frac{\alpha}{2}(\rho_j+\lambda_j-\mathrm{g}_2)}{\cos\frac{\alpha}{2}(\rho_j+\lambda_j)}<\infty.
\end{equation*}
Finally, we see from the inequalities
\begin{gather*}
-\tfrac{\pi}{\alpha}<\rho_n-\mathrm{g}_2-1\leq\rho_j+\lambda_j-\mathrm{g}_2-1<\rho_j+\lambda_j-\mathrm{g}'_2-\tfrac{1}{2}<\rho_j+\lambda_j+\mathrm{g}_2\leq\tfrac{\pi}{\alpha},\\
-\tfrac{\pi}{\alpha}<-\tfrac{1}{2}<\rho_j+\lambda_j-\tfrac{1}{2}<\rho_1+m<\tfrac{\pi}{\alpha}
\end{gather*}
that
\begin{equation*}
0<\tfrac{\cos\frac{\alpha}{2}(\rho_j+\lambda_j-\mathrm{g}_2'-\frac{1}{2})}{\cos\frac{\alpha}{2}(\rho_j+\lambda_j-\frac{1}{2})}<\infty.
\end{equation*}
To summarize, we have shown that $0\leq B_{\lambda,-j}<\infty$ with the vanishing occurring iff $\lambda_j=0$ or $j<n$ and $\lambda_j=\lambda_{j+1}$, i.e.  iff $\lambda-e_j\notin\Lambda^{(n,m)}$.
\end{proof}

Lemma \ref{lemma:truncation} ensures that---for parameters subject to the conditions in Eqs. \eqref{coupling-conditions}, \eqref{truncation}---the $\binom{n+m}{n}$-dimensional subspace
\begin{equation}\label{lattice-functions2}
\mathcal{C}(\Lambda^{(n,m)})=\{\lambda\overset{f}{\to} f_\lambda\in\mathbb{C}\mid\lambda\in\Lambda^{(n,m)}\}
\end{equation}
of $\mathcal{C}(\Lambda^{(n)})$ is stable with respect to the action of $H$ \eqref{Hamiltonian}--\eqref{C-r}:
\begin{equation}\label{Hamiltonian2}
(Hf)_\lambda=A_\lambda f_\lambda+\sum_{\substack{1\leq j\leq n,\,\varepsilon=\pm 1\\\lambda+\varepsilon e_j\in\Lambda^{(n,m)}}}B_{\lambda,\varepsilon j}f_{\lambda+\varepsilon e_j}\qquad(f\in\mathcal{C}(\Lambda^{(n,m)}),\ \lambda\in\Lambda^{(n,m)}).
\end{equation}
In addition, we also see with the aid of Lemma \ref{lemma:truncation} that for parameters in the domain given by Eqs. \eqref{coupling-conditions} and \eqref{truncation}  the coefficients $A_\lambda$ \eqref{A-lambda} and $B_{\lambda,\varepsilon j}$ \eqref{B-lambda} in Eq. \eqref{Hamiltonian2} remain regular provided $\mathrm{g}\neq 1$.

\section{Hilbert space}
\label{sec:3}

\subsection{Inner product}
\label{subsec:3.1}

We put \emph{elliptic weights} on the lattice $\Lambda^{(n,m)}$ by introducing the following function $\Delta\in\mathcal{C}(\Lambda^{(n,m)})$:
\begin{equation}\label{weights}
\Delta_\lambda=\prod_{\substack{1\leq j\leq n\\1\leq r\leq 4}}\tfrac{[\rho_j+1,\rho_j+\mathrm{g}_r,\rho_j+\mathrm{g}'_r+\frac{1}{2}]_{r,\lambda_j}}{[\rho_j,\rho_j+1-\mathrm{g}_r,\rho_j-\mathrm{g}'_r+\frac{1}{2}]_{r,\lambda_j}}\prod_{\substack{1\leq j<k\leq n\\\delta=\pm 1}}\tfrac{[\rho_j+\delta\rho_k+\mathrm{g},\rho_j+\delta\rho_k+1]_{1,\lambda_j+\delta\lambda_k}}{[\rho_j+\delta\rho_k,\rho_j+\delta\rho_k+1-\mathrm{g}]_{1,\lambda_j+\delta\lambda_k}},
\end{equation}
where $\lambda\in\Lambda^{(n,m)}$ and $[z_1,\dots,z_N]_{r,l}$ denotes the elliptic shifted factorial
\begin{equation}
\label{elliptic-factorials}
[z_1,\dots,z_N]_{r,l}=\prod_{\substack{1\leq j\leq N\\0\leq k<l}}[z_j+k]_r\quad\text{with}\ [z]_{r,0}=1
\end{equation}
(for $N\in\mathbb{N}$ and $l=0,1,2,\dots$). It is clear from this definition that $[\cdot]_{r,l}$ has the factorial property, i.e.
\begin{equation}\label{factorial-property}
[z]_{r,l+1}=[z]_{r,l}[z+l]_r,\quad l\geq 0.
\end{equation}
With the aid of the factorial property \eqref{factorial-property} and the duplication formula $[2z]_1=2\prod_{1\leq r\leq 4}[z]_r$ (see e.g. \cite[\S20.7(iii)]{olv-loz-boi-cla:nist}), the elliptic weight function $\Delta_\lambda$ \eqref{weights} can be rewritten as follows
\begin{equation}\label{Delta2}
\begin{split}
\Delta_\lambda=&\prod_{1\leq j\leq n}\tfrac{[2\rho_j+2\lambda_j]_1}{[2\rho_j]_1}\prod_{1\leq r\leq 4}\tfrac{[\rho_j+\mathrm{g}_r,\rho_j+\mathrm{g}'_r+\frac{1}{2}]_{r,\lambda_j}}{[\rho_j+1-\mathrm{g}_r,\rho_j-\mathrm{g}'_r+\frac{1}{2}]_{r,\lambda_j}}\\
&\times \prod_{\substack{1\leq j<k\leq n\\\delta=\pm 1}}\tfrac{[\rho_j+\delta\rho_k+\lambda_j+\delta\lambda_k]_1}{[\rho_j+\delta\rho_k]_1}\tfrac{[\rho_j+\delta\rho_k+\mathrm{g}]_{1,\lambda_j+\delta\lambda_k}}{[\rho_j+\delta\rho_k+1-\mathrm{g}]_{1,\lambda_j+\delta\lambda_k}}.
\end{split}
\end{equation}
The following lemma checks that the weight function $\Delta$ \eqref{weights} is \emph{positive}.

\begin{lemma}[Positivity]\label{lemma:positivity}
For parameters subject to the conditions in Eqs. \eqref{coupling-conditions}, \eqref{truncation} and any $\lambda\in\Lambda^{(n,m)}$,
the elliptic weight $\Delta_\lambda$ \eqref{weights} is \emph{positive}.
\end{lemma}

\begin{proof}
The positivity of $\Delta_\lambda$ follows by checking that all factors on the RHS of Eq. \eqref{weights} remain positive. To see this, let us list the following inequalities for  $1\leq j<k\leq n$ and $0\leq l<\lambda_j+\lambda_k$:
\begin{gather*}
0<\mathrm{g}<\rho_j-\rho_k+l+\mathrm{g}<\rho_j+\rho_k+l+\mathrm{g}<2\rho_1+2m+\mathrm{g}<\tfrac{2\pi}{\alpha},\\
0<1<\rho_j-\rho_k+l+1<\rho_j+\rho_k+l+1\leq 2\rho_1+2m<\tfrac{2\pi}{\alpha},\\
0<\mathrm{g}\leq\rho_j-\rho_k+l<\rho_j+\rho_k+l<2\rho_1+2m<\tfrac{2\pi}{\alpha},\\
0<1\leq\rho_j-\rho_k+l+1-\mathrm{g}<\rho_j+\rho_k+l+1-\mathrm{g}<2\rho_1+2m<\tfrac{2\pi}{\alpha}.
\end{gather*}
It thus follows that
\begin{equation*}
0<\tfrac{[\rho_j+\delta\rho_k+\mathrm{g},\rho_j+\delta\rho_k+1]_{1,\lambda_j+\delta\lambda_k}}{[\rho_j+\delta\rho_k,\rho_j+\delta\rho_k+1-\mathrm{g}]_{1,\lambda_j+\delta\lambda_k}}<\infty
\end{equation*}
for all $1\leq j<k\leq n$ and $\delta=\pm 1$. Similarly, one has that
\begin{gather*}
0<\rho_n\leq\rho_j+l<\rho_j+1+l\leq\rho_1+m<\tfrac{\pi}{\alpha},\\
0\leq\rho_j-\mathrm{g}_1+l<\rho_j\pm\mathrm{g}'_1+\tfrac{1}{2}+l<\rho_j+\mathrm{g}_1+1+l\leq\rho_1+\mathrm{g}_1+m<\tfrac{2\pi}{\alpha},\\
0<1\leq\rho_j+1-\mathrm{g}_1+l\leq\rho_1+m-\mathrm{g}_1<\tfrac{\pi}{\alpha}
\end{gather*}
for $1\leq j\leq n$ and $0\leq l<\lambda_j$, which implies that
\begin{equation*}
0<\tfrac{[\rho_j+1,\rho_j+\mathrm{g}_1,\rho_j+\mathrm{g}'_1+\frac{1}{2}]_{1,\lambda_j}}{[\rho_j,\rho_j+1-\mathrm{g}_1,\rho_j-\mathrm{g}'_1+\frac{1}{2}]_{1,\lambda_j}}<\infty
\end{equation*}
for all $1\leq j\leq n$. Finally, if $1\leq j\leq n$ and $0\leq l<\lambda_j$ then
\begin{gather*}
-\tfrac{\pi}{\alpha}<\rho_n-\mathrm{g}_2\leq\rho_j-\mathrm{g}_2+l<\rho_j\pm\mathrm{g}'_2+\tfrac{1}{2}+l<\rho_j+\mathrm{g}_2+1+l\leq\rho_1+\mathrm{g}_2+m=\tfrac{\pi}{\alpha},\\
-\tfrac{\pi}{\alpha}<\rho_n+1-\mathrm{g}_2\leq\rho_j+1-\mathrm{g}_2+l\leq\rho_1+m-\mathrm{g}_2<\tfrac{\pi}{\alpha},
\end{gather*}
so
\begin{equation*}
0<\tfrac{[\rho_j+1,\rho_j+\mathrm{g}_2,\rho_j+\mathrm{g}'_2+\frac{1}{2}]_{2,\lambda_j}}{[\rho_j,\rho_j+1-\mathrm{g}_2,\rho_j-\mathrm{g}'_2+\frac{1}{2}]_{2,\lambda_j}}<\infty
\end{equation*}
for all $1\leq j\leq n$.
\end{proof}

Lemma \ref{lemma:positivity} confirms that $\Delta$ \eqref{weights} constitutes a discrete weight function  endowing $\mathcal{C}(\Lambda^{(n,m)})$ with an inner product:
\begin{equation}\label{inner-product}
\langle f,g\rangle_\Delta=\sum_{\lambda\in\Lambda^{(n,m)}}f_\lambda \overline{g_\lambda}\Delta_\lambda\qquad(f,g\in\mathcal{C}(\Lambda^{(n,m)})).
\end{equation}
This promotes $\mathcal{C}(\Lambda^{(n,m)})$ to an $\binom{n+m}{n}$-dimensional Hilbert space that will be denoted by $\ell^2(\Lambda^{(n,m)},\Delta)$.

\subsection{Self-adjointness}
\label{subsec:3.2}

The elliptic weights $\Delta_\lambda$ \eqref{weights} satisfy a recurrence relation that involves the coefficients $B_{\lambda,\varepsilon j}$ \eqref{B-lambda}.

\begin{lemma}[Recurrence for Elliptic Weights]\label{lemma:recurrence}
If $\lambda\in\Lambda^{(n,m)}$, $j\in\{1,\dots,n\}$ and $\varepsilon\in\{\pm 1\}$ are such that $\mu=\lambda+\varepsilon e_j\in\Lambda^{(n,m)}$, then the elliptic weights \eqref{weights} satisfy the following recurrence relation
\begin{equation}\label{recurrence}
B_{\lambda,\varepsilon j}\Delta_\lambda=B_{\mu,-\varepsilon j}\Delta_\mu.
\end{equation}
\end{lemma}

\begin{proof}
The key observation is that $\Delta_\mu$ and $\Delta_\lambda$ only differ in factors involving the index $j$. To those the factorial property \eqref{factorial-property} can be applied, so as to split $\Delta_\mu$ up into a product of $\Delta_\lambda$ and extra factors. Upon close inspection, these extra factors are seen to make up $B_{\lambda,\varepsilon j}/B_{\mu,-\varepsilon j}$. Specifically, we have that
\begin{equation*}
\begin{split}
\Delta_\mu=&\prod_{\substack{1\leq i\leq n\\1\leq r\leq 4}}\tfrac{[\rho_i+1,\rho_i+\mathrm{g}_r,\rho_i+\mathrm{g}'_r+\frac{1}{2}]_{r,\mu_i}}{[\rho_i,\rho_i+1-\mathrm{g}_r,\rho_i-\mathrm{g}'_r+\frac{1}{2}]_{r,\mu_i}}\prod_{\substack{1\leq i<k\leq n\\\delta=\pm 1}}\tfrac{[\rho_i+\delta\rho_k+\mathrm{g},\rho_i+\delta\rho_k+1]_{1,\mu_i+\delta\mu_k}}{[\rho_i+\delta\rho_k,\rho_i+\delta\rho_k+1-\mathrm{g}]_{1,\mu_i+\delta\mu_k}}\\
=&\prod_{\substack{1\leq i\leq n\\1\leq r\leq 4}}\tfrac{[\rho_i+1,\rho_i+\mathrm{g}_r,\rho_i+\mathrm{g}'_r+\frac{1}{2}]_{r,\lambda_i}}{[\rho_i,\rho_i+1-\mathrm{g}_r,\rho_i-\mathrm{g}'_r+\frac{1}{2}]_{r,\lambda_i}}\prod_{\substack{1\leq i<k\leq n\\\delta=\pm 1}}\tfrac{[\rho_i+\delta\rho_k+\mathrm{g},\rho_i+\delta\rho_k+1]_{1,\lambda_i+\delta\lambda_k}}{[\rho_i+\delta\rho_k,\rho_i+\delta\rho_k+1-\mathrm{g}]_{1,\lambda_i+\delta\lambda_k}}\\
&\times\prod_{1\leq r\leq 4}\tfrac{[\rho_j+\lambda_j+\varepsilon]_r[\rho_j+\lambda_j+\varepsilon\mathrm{g}_r]_r[\rho_j+\lambda_j+\varepsilon(\mathrm{g}'_r+\frac{1}{2})]_r}{[\rho_j+\lambda_j]_r[\rho_j+\lambda_j+\varepsilon(1-\mathrm{g}_r)]_r[\rho_j+\lambda_j-\varepsilon(\mathrm{g}'_r-\frac{1}{2})]_r}\\
&\times\prod_{\substack{j<k\leq n\\\delta=\pm 1}}\tfrac{[\rho_j+\delta\rho_k+\lambda_j+\delta\lambda_k+\varepsilon\mathrm{g}]_1[\rho_j+\delta\rho_k+\lambda_j+\delta\lambda_k+\varepsilon]_1}{[\rho_j+\delta\rho_k+\lambda_j+\delta\lambda_k]_1[\rho_j+\delta\rho_k+\lambda_j+\delta\lambda_k+\varepsilon(1-\mathrm{g})]_1}\\
&\times\prod_{\substack{1\leq i<j\\\delta=\pm 1}}\tfrac{[\rho_i+\delta\rho_j+\lambda_i+\delta\lambda_j+\varepsilon\delta\mathrm{g}]_1[\rho_i+\delta\rho_j+\lambda_i+\delta\lambda_j+\varepsilon\delta]_1}{[\rho_i+\delta\rho_j+\lambda_i+\delta\lambda_j]_1[\rho_i+\delta\rho_j+\lambda_i+\delta\lambda_j+\varepsilon\delta(1-\mathrm{g})]_1}\\
=&\Delta_\lambda\prod_{1\leq r\leq 4}\tfrac{[\rho_j+\lambda_j+\varepsilon\mathrm{g}_r]_r[\rho_j+\lambda_j+\varepsilon(\mathrm{g}'_r+\frac{1}{2})]_r}{[\rho_j+\lambda_j]_r[\rho_j+\lambda_j+\frac{\varepsilon}{2}]_r}\tfrac{[\rho_j+\mu_j]_r[\rho_j+\mu_j-\frac{\varepsilon}{2}]_r}{[\rho_j+\mu_j-\varepsilon\mathrm{g}_r]_r[\rho_j+\mu_j-\varepsilon(\mathrm{g}'_r+\frac{1}{2})]_r}\\
&\times\prod_{\substack{1\leq k\leq n\\k\neq j,\delta=\pm 1}}\tfrac{[\rho_j+\delta\rho_k+\lambda_j+\delta\lambda_k+\varepsilon\mathrm{g}]_1}{[\rho_j+\delta\rho_k+\lambda_j+\delta\lambda_k]_1}\tfrac{[\rho_j+\delta\rho_k+\mu_j+\delta\mu_k]_1}{[\rho_j+\delta\rho_k+\mu_j+\delta\mu_k-\varepsilon\mathrm{g}]_1}\\
=&\Delta_\lambda\frac{B_{\lambda,\varepsilon j}}{B_{\mu,-\varepsilon j}}.
\end{split}
\end{equation*}
\end{proof}

The self-adjointness of $H$ \eqref{Hamiltonian2} in $\ell^2(\Lambda^{(n,m)},\Delta)$ now follows
as an immediate consequence of the recurrence in Lemma \ref{lemma:recurrence}.

\begin{proposition}[Self-adjointness]\label{proposition:self-adjointness}
For parameters subject to the conditions in Eqs. \eqref{coupling-conditions}, \eqref{truncation} (and $\mathrm{g}\neq 1$), the
difference operator $H$ \eqref{Hamiltonian2} is \emph{self-adjoint} in the Hilbert space $\ell^2(\Lambda^{(n,m)},\Delta)$, i.e.
\begin{equation}\label{self-adjointness}
\forall f,g\in\ell^2(\Lambda^{(n,m)},\Delta):\qquad\langle Hf,g\rangle_\Delta=\langle f,Hg\rangle_\Delta.
\end{equation}
\end{proposition}

\begin{proof}
A straightforward calculation reveals that for all $f,g\in\ell^2(\Lambda^{(n,m)},\Delta)$ we have
\begin{equation*}
\begin{split}
\langle Hf,g\rangle_\Delta&=\sum_{\lambda\in\Lambda^{(n,m)}}(H f)_\lambda\overline{g_\lambda}\Delta_\lambda\\
&=\sum_{\lambda\in\Lambda^{(n,m)}}A_\lambda f_\lambda\overline{g_\lambda}\Delta_\lambda+\sum_{\substack{1\leq j\leq n\\\varepsilon=\pm 1}}\sum_{\substack{\lambda\in\Lambda^{(n,m)}\\\lambda+\varepsilon e_j\in\Lambda^{(n,m)}}}B_{\lambda,\varepsilon j}f_{\lambda+\varepsilon e_j}\overline{g_\lambda}\Delta_\lambda\\
&\overset{\eqref{recurrence}}{=}\sum_{\lambda\in\Lambda^{(n,m)}} f_\lambda\overline{A_\lambda g_\lambda}\Delta_\lambda+\sum_{\substack{1\leq j\leq n\\\varepsilon=\pm 1}}\sum_{\substack{\lambda\in\Lambda^{(n,m)}\\\lambda+\varepsilon e_j\in\Lambda^{(n,m)}}}B_{\lambda+\varepsilon e_j,-\varepsilon j}f_{\lambda+\varepsilon e_j}\overline{g_\lambda}\Delta_{\lambda+\varepsilon e_j}\\
&=\sum_{\kappa\in\Lambda^{(n,m)}} f_\kappa\overline{A_\kappa g_\kappa}\Delta_\kappa+\sum_{\substack{1\leq j\leq n\\\varepsilon=\pm 1}}\sum_{\substack{\kappa-\varepsilon e_j\in\Lambda^{(n,m)}\\\kappa\in\Lambda^{(n,m)}}}B_{\kappa,-\varepsilon j}f_\kappa\overline{g_{\kappa-\varepsilon e_j}}\Delta_\kappa\\
&=\sum_{\kappa\in\Lambda^{(n,m)}}f_\kappa\overline{(Hg)_\kappa}\Delta_\kappa=\langle f,Hg\rangle_\Delta.
\end{split}
\end{equation*}
\end{proof}

\section{Diagonalization}
\label{sec:4}

\subsection{Spectrum}
\label{subsec:4.1}

Consider the following scaled variants of the elementary trigonometric functions
\begin{equation}\label{key}
[z]_{1,q}=\tfrac{\sin(\frac{\alpha}{2}z)}{\sin(\frac{\alpha}{2})}=\tfrac{q^{\frac{z}{2}}-q^{-\frac{z}{2}}}{q^{\frac{1}{2}}-q^{-\frac{1}{2}}},\quad
[z]_{2,q}=\cos(\tfrac{\alpha}{2}z)=\tfrac{q^{\frac{z}{2}}+q^{-\frac{z}{2}}}{2}\quad\text{with}\ q=e^{\mathrm{i}\alpha},
\end{equation}
together with the corresponding shifted factorials
\begin{equation}
\label{q-factorials}
[z_1,\dots,z_N]_{r,q,l}=\prod_{\substack{1\leq j\leq N\\0\leq k<l}}[z_j+k]_{r,q}\quad\text{with}\ [z]_{r,q,0}=1.
\end{equation}
It is immediate from Eqs. \eqref{thetas}--\eqref{theta4} that $[z;p]_r$ extends analytically in $p$ to the interval $-1<p<1$ such that
\begin{equation}\label{thetas-q}
[z;0]_r=[z]_{r,q}\quad(r=1,2)\quad\text{and}\quad[z;0]_3=[z;0]_4=1.
\end{equation}

\begin{proposition}[Eigenvalues]\label{proposition:simple-spectrum}
\emph{(i)} For parameters subject to the conditions in Eqs. \eqref{coupling-conditions}, \eqref{truncation} and $\mathrm{g}\neq 1$, the eigenvalues of the difference operator $H$ \eqref{Hamiltonian2}
are given by \emph{real-analytic} functions $E_\nu$, $\nu\in\Lambda^{(n,m)}$ in $p\in (-1,1)$ that specialize at $p=0$ to
\begin{equation}\label{trig-eigenvalues}
E_\nu\vert_{p=0}=
2\sum_{1\leq j\leq n}[2(\hat\rho_j+\nu_j)]_{2,q} ,
\end{equation}
where $\hat\rho_j=(n-j)\mathrm{g}+\frac{1}{2}(\mathrm{g}_1+\mathrm{g}_2+\mathrm{g}'_1+\mathrm{g}'_2)$,  $j=1,\dots,n$.

\emph{(ii)} For generic coupling values in the indicated domain, the eigenvalues $E_\nu$, $\nu\in\Lambda^{(n,m)}$ from part \emph{(i)} are \emph{distinct} as analytic functions of $p\in (-1,1)$.
\end{proposition}

\begin{proof}
Since the difference operator $H$ \eqref{Hamiltonian2} is self-adjoint in $\ell^2(\Lambda^{(n,m)},\Delta)$ by Proposition \ref{proposition:self-adjointness},
with coefficients $A_\lambda$ \eqref{A-lambda} and $B_{\lambda,\varepsilon j}$ \eqref{B-lambda} that are 
analytic in $p\in (-1,1)$, it follows that the spectrum is given by real eigenvalues that are also analytic in $p\in (-1,1)$ (cf. e.g. \cite[Chapter II, Theorem 6.1]{kat:perturbation}).

Furthermore, it is clear from Eqs. \eqref{A-lambda}--\eqref{C-r} and Eq. \eqref{thetas-q} that at $p=0$
\begin{equation*}
A_\lambda\vert_{p=0}=\sum_{r=1,2}\mathrm{c}_{r,q}\Bigl( \prod_{\substack{1\leq j\leq n\\\delta=\pm 1}}\tfrac{[\delta(\rho_j+\lambda_j)-\frac{1}{2}+\mathrm{g}]_{r,q}}{[\delta(\rho_j+\lambda_j)-\frac{1}{2}]_{r,q}}-1\Bigr) ,
\end{equation*}
and
\begin{equation*}
B_{\lambda,\varepsilon j}\vert_{p=0}=\prod_{r=1,2}\tfrac{[\rho_j+\lambda_j+\varepsilon\mathrm{g}_r]_{r,q}[\rho_j+\lambda_j+\varepsilon(\mathrm{g}'_r+\frac{1}{2})]_{r,q}}{[\rho_j+\lambda_j]_{r,q}[\rho_j+\lambda_j+\frac{\varepsilon}{2}]_{r,q}}\prod_{\substack{1\leq k\leq n\\k\neq j,\delta=\pm 1}}\tfrac{[\rho_j+\delta\rho_k+\lambda_j+\delta\lambda_k+\varepsilon\mathrm{g}]_{1,q}}{[\rho_j+\delta\rho_k+\lambda_j+\delta\lambda_k]_{1,q}}
\end{equation*}
with
\begin{equation*}
\mathrm{c}_{r,q}=\tfrac{2}{[\mathrm{g}]_{1,q}[\mathrm{g}-1]_{1,q}}\prod_{s=1,2}[\mathrm{g}_{\pi_r(s)}-\tfrac{1}{2}]_{s,q}[\mathrm{g}'_{\pi_r(s)}]_{s,q}.
\end{equation*}
By means of the functional identity (cf. \cite[Section III.C]{die:difference}) 
\begin{equation*}
A_\lambda\vert_{p=0}+\sum_{\substack{1\leq j\leq n\\\varepsilon=\pm 1}}B_{\lambda,\varepsilon j}\vert_{p=0}=c,
\end{equation*}
where $ c=2\sum_{1\leq j\leq n}[2\hat\rho_j]_{2,q} $ and 
$\hat\rho_j=(n-j)\mathrm{g}+\frac{1}{2}(\mathrm{g}_1+\mathrm{g}_2+\mathrm{g}'_1+\mathrm{g}'_2)$, it thus follows
via Lemma \ref{lemma:truncation} that for all $f\in\ell^2(\Lambda^{(n,m)},\Delta)$ and $\lambda\in\Lambda^{(n,m)}$:
\begin{equation*}
(Hf)_\lambda\vert_{p=0}=\sum_{\substack{1\leq j\leq n\\\lambda+\varepsilon e_j\in\Lambda^{(n,m)}}}B_{\lambda,\varepsilon j}\vert_{p=0}(f_{\lambda+\varepsilon e_j}-f_\lambda)+c.
\end{equation*}
The upshot is that  $H$ \eqref{Hamiltonian2} reduces at $p=0$ to the
discrete Koornwinder-Macdonald operator $D^{\mathrm{qR}}$ from \cite[Eq. (6.15)]{die-sto:multivariable} up to an additive constant:
\begin{equation*}
H\vert_{p=0}=D^{\mathrm{qR}}+c.
\end{equation*}
By \cite[Theorem 6.5]{die-sto:multivariable},  
the eigenvalues of the discrete Koornwinder-Macdonald operator $D^{\mathrm{qR}}$ are given by
\begin{equation*}
2\sum_{1\leq j\leq n}([2(\hat\rho_j+\nu_j)]_{2,q}-[2\hat\rho_j]_{2,q})\qquad (\nu\in\Lambda^{(n,m)}),
\end{equation*}
so the eigenvalues $E_\nu\vert_{p=0}$ of $H\vert_{p=0}$ become as asserted in Eq. \eqref{trig-eigenvalues},
which completes the proof of part \emph{(i)}.

To infer part \emph{(ii)}, we observe from the explicit formula that---apart from a multiplicative factor of the form $-q^{-\hat{\rho}_1-m}$---the eigenvalues $E_\nu\vert_{p=0}$, $\nu\in\Lambda^{(n,m)}$
are of the form
\begin{equation*}
 \sum_{1\leq j \leq n}  \left( q^{(j-1)\mathrm{g}+\mathrm{g}_1^\prime+\mathrm{g}_2^\prime+ \nu_{n+1-j}} -q^{(j-1)\mathrm{g}+m-\nu_j} \right)
 \quad\text{with}\ q=\exp\bigl( {\textstyle \frac{\pi \mathrm{i}}{(n-1)\mathrm{g}+\mathrm{g}_1+\mathrm{g}_2+m} }\bigr) 
\end{equation*}
(where we have also used that $q^{(n-1)\mathrm{g}+\mathrm{g}_1+\mathrm{g}_2+m}=-1$). At $p=0$ the eigenvalues  are therefore distinct as analytic functions of the coupling parameters, so for generic coupling values in our domain the eigenvalues $E_\nu$, $\nu\in\Lambda^{(n,m)}$ of $H$ \eqref{Hamiltonian2} must be distinct as analytic functions of $p\in (-1,1)$.
\end{proof}

\subsection{Eigenbasis}
\label{subsec:4.2}

Starting from the function $\chi\in\ell^2(\Lambda^{(n,m)},\Delta)$ given by
\begin{subequations}
\begin{equation}\label{chi}
\chi_\lambda=\begin{cases}1,&\text{if}\ \lambda=0,\\0,&\text{if}\ \lambda\neq 0,\end{cases}\qquad(\lambda\in\Lambda^{(n,m)}),
\end{equation}
we will construct eigenfunctions $h^{(\nu)}$, $\nu\in\Lambda^{(n,m)}$ in $\ell^2(\Lambda^{(n,m)},\Delta)$ as follows:
\begin{equation}\label{h-nu}
h^{(\nu)}=\big(\prod_{\substack{\mu\in\Lambda^{(n,m)}\\\mu\neq\nu}}\tfrac{H-E_\mu}{E_\nu-E_\mu}\big)\chi.
\end{equation}
\end{subequations}
By Proposition \ref{proposition:simple-spectrum}, the function values $h^{(\nu)}_\lambda$, $\lambda\in\Lambda^{(n,m)}$ are well-defined
as meromorphic functions
of $p\in (-1,1)$ at the generic values of the coupling parameters from the domain \eqref{coupling-conditions}, \eqref{truncation} for which the discriminant
\begin{equation}\label{discriminant}
\Delta(H)=\prod_{\substack{\mu,\nu\in\Lambda^{(n,m)}\\ \mu\neq\nu} } (E_\nu-E_\mu) 
\end{equation}
does not vanish as an analytic function of $p\in (-1,1)$.

\begin{theorem}[Eigenfunctions]\label{theorem:diagonalization}
The following statements hold for parameters subject to the conditions in Eqs. \eqref{coupling-conditions}, \eqref{truncation} with $\mathrm{g}\neq 1$.

\begin{subequations}
\emph{(i)} The second-order difference operator $H$ \eqref{Hamiltonian2}
is diagonalized in the Hilbert space $\ell^2(\Lambda^{(n,m)},\Delta)$ by an orthonormal basis of eigenfunctions
$f^{(\nu)}$, $\nu\in\Lambda^{(n,m)}$ that depend analytically on $p\in (-1,1)$ such that
$Hf^{(\nu)}=E_\nu f^{(\nu)} $ with $E_\nu$ as in Proposition \ref{proposition:simple-spectrum}.

\emph{(ii)} For generic values of the coupling parameters such that the discriminant
$\Delta(H)$ \eqref{discriminant} does not vanish
as an analytic function of $p\in (-1,1)$ (cf. Proposition \ref{proposition:simple-spectrum}), the values of
$h^{(\nu)}$ \eqref{h-nu} extend to real-analytic functions of $p\in (-1,1)$
such that $\forall \nu,\tilde{\nu}\in\Lambda^{(n,m)}$:
\begin{equation}\label{eigenvalue-equation}
Hh^{(\nu)}=E_\nu h^{(\nu)} 
\end{equation}
and
\begin{equation}\label{orthogonality}
\langle h^{(\nu)},h^{(\tilde\nu)}\rangle_\Delta=
\begin{cases}
h^{(\nu)}_0 &\text{if}\ \nu=\tilde\nu,\\
0&\text{if}\ \nu\neq\tilde\nu .
\end{cases}
\end{equation}
Moreover, the functions $h^{(\nu)}\in \ell^2(\Lambda^{(n,m)},\Delta)$ thus obtained
coincide with the orthonormal eigenfunctions in part \emph{(i)} up to normalization:
\begin{equation}\label{eigenfunctions}
h^{(\nu)}=\overline{f^{(\nu)}_0} f^{(\nu)} .
\end{equation}

\emph{(iii)} The extension of the functions $h^{(\nu)}$, $\nu\in\Lambda^{(n,m)}$ to the case of arbitrary coupling parameters stemming from Eq. \eqref{eigenfunctions} constitutes an orthogonal eigenbasis for  $H$ \eqref{Hamiltonian2}
such that $Hh^{(\nu)}=E_\nu h^{(\nu)} $, unless $p\in (-1,1)$ belongs to the finite zero locus
$\{ 0< |p| \leq \epsilon \mid \prod_{\nu\in\Lambda^{(n,m)}} h^{(\nu)}_0 =0\}$ for some $\epsilon\in (0,1)$.

\emph{(iv)} At $p=0$ the eigenbasis in part \emph{(iii)} is given explicitly by
\begin{equation}\label{trig-eigenfunctions}
h_\lambda^{(\nu)}\vert_{p=0}=\frac{\textsc{c}_{\lambda,q}}{\textsc{n}_{\nu ,q} } P_\lambda(q^{\hat\rho+\nu};q,t,{\mathrm{a}}, {\mathrm{b}},{\mathrm{c}},{\mathrm{d}})
\quad (\lambda,\nu\in\Lambda^{(n,m)}),
\end{equation}
where
\begin{equation}\label{trig-c-lambda}
\textsc{c}_{\lambda,q}=\prod_{\substack{1\leq j\leq n\\r=1,2}}\tfrac{[\rho_j,\rho_j+\frac{1}{2}]_{r,q,\lambda_j}}{[\rho_j+\mathrm{g}_r,\rho_j+\mathrm{g}'_r+\frac{1}{2}]_{r,q,\lambda_j}}\prod_{\substack{1\leq j<k\leq n\\\delta=\pm 1}}\tfrac{[\rho_j+\delta\rho_k]_{1,q,\lambda_j+\delta\lambda_k}}{[\rho_j+\delta\rho_k+\mathrm{g}]_{1,q,\lambda_j+\delta\lambda_k}}
\end{equation}
and $P_\lambda$ denotes the monic  Koornwinder-Macdonald polynomial \cite{koo:askey} with parameters
\begin{equation}\label{trig-parameters}
q=e^{\mathrm{i}\alpha},\quad t=q^{\mathrm{g}},\quad {\mathrm{a}}=q^{\hat{\mathrm{g}}_1},\quad {\mathrm{b}}=-q^{\hat{\mathrm{g}}_2},\quad {\mathrm{c}}=q^{\hat{\mathrm{g}}'_1+\frac{1}{2}},\quad {\mathrm{d}}=-q^{\hat{\mathrm{g}}'_2+\frac{1}{2}} .
\end{equation}
Here the transformed parameters $\hat{\mathrm{g}}_1,\hat{\mathrm{g}}_2,\hat{\mathrm{g}}'_1,\hat{\mathrm{g}}'_2$ are related to $\mathrm{g}_1,\mathrm{g}_2,\mathrm{g}'_1,\mathrm{g}'_2$ via the reflection
\begin{equation}\label{dual-couplings}
\begin{pmatrix}
\hat{\mathrm{g}}_1\\\hat{\mathrm{g}}_2\\\hat{\mathrm{g}}'_1\\\hat{\mathrm{g}}'_2
\end{pmatrix}=\frac{1}{2}
\begin{pmatrix}
1&\phantom{+}1&\phantom{+}1&\phantom{+}1\\
1&\phantom{+}1&-1&-1\\
1&-1&\phantom{+}1&-1\\
1&-1&-1&\phantom{+}1
\end{pmatrix}
\begin{pmatrix}
\mathrm{g}_1\\\mathrm{g}_2\\\mathrm{g}'_1\\\mathrm{g}'_2
\end{pmatrix},
\end{equation}
and
\begin{equation}\label{snorms:p=0}
\textsc{n}_{\nu ,q} = \sum_{\lambda\in\Lambda^{(n,m)} }
\textsc{c}_{\lambda,q}^2
P_\lambda^2 (q^{\hat\rho+\nu};q,t,{\mathrm{a}},{\mathrm{b}},{\mathrm{c}},{\mathrm{d}}) \Delta_{\lambda,q} ,
\end{equation}
with
\begin{equation}\label{Delta:p=0}
\begin{split}
\Delta_{\lambda,q}=&\prod_{1\leq j\leq n}\tfrac{[2\rho_j+2\lambda_j]_{1,q}}{[2\rho_j]_{1,q}}\prod_{r=1,2}
\tfrac{[\rho_j+\mathrm{g}_r,\rho_j+\mathrm{g}'_r+\frac{1}{2}]_{r,q,\lambda_j}}{[\rho_j+1-\mathrm{g}_r,\rho_j-\mathrm{g}'_r+\frac{1}{2}]_{r,q,\lambda_j}}\\
&\times\prod_{\substack{1\leq j<k\leq n\\\delta=\pm 1}}\tfrac{[\rho_j+\delta\rho_k+\lambda_j+\delta\lambda_k]_{1,q}}{[\rho_j+\delta\rho_k]_{1,q}}\tfrac{[\rho_j+\delta\rho_k+\mathrm{g}]_{1,q,\lambda_j+\delta\lambda_k}}{[\rho_j+\delta\rho_k+1-\mathrm{g}]_{1,q,\lambda_j+\delta\lambda_k}}.
\end{split}
\end{equation}
\end{subequations}
\end{theorem}

\begin{proof}
\emph{(i)} 
By Proposition \ref{proposition:self-adjointness}, the difference operator $H$ \eqref{Hamiltonian2} is self-adjoint
in the $\binom{n+m}{n}$-dimensional Hilbert space $\ell^2(\Lambda^{(n,m)},\Delta)$. 
The spectral theorem for self-adjoint operators in finite dimension (cf. e.g. \cite[Chapter I.6.9]{kat:perturbation}) thus guarantees the existence of an orthonormal basis of eigenfunctions.
Since the dependence on $p$ of the coefficients $A_\lambda$ \eqref{A-lambda} and $B_{\lambda,\varepsilon j}$ \eqref{B-lambda} of the difference operator as well as that of the orthogonality measure $\Delta$ \eqref{Delta2}
are analytic for $p\in (-1,1)$, the eigenfunctions constituting the orthonormal
eigenbasis too can be chosen analytically in $p\in (-1,1)$ in one-to-one correspondence with the (not necessarily multiplicity-free) eigenvalues of
Proposition \ref{proposition:simple-spectrum} (cf. e.g. \cite[Chapter II.6.2]{kat:perturbation}).

\emph{(ii)} Upon decomposing $\chi$ \eqref{chi} in the orthonormal eigenbasis $f^{(\nu)}$, $\nu\in\Lambda^{(n,m)}$ from
part \emph{(i)}:
\begin{equation*}
\chi=\sum_{\nu\in\Lambda^{(n,m)}} \langle\chi,f^{(\nu)}\rangle_\Delta f^{(\nu)}=
\sum_{\nu\in\Lambda^{(n,m)}} \overline{f_0^{(\nu)}} f^{(\nu)}
\end{equation*}
(since $\Delta_0=1$), it is seen that with our genericity assumptions on the coupling parameters in place
\begin{align*}
h^{(\nu)}&=\big(\prod_{\substack{\mu\in\Lambda^{(n,m)}\\\mu\neq\nu}}\tfrac{H-E_\mu}{E_\nu-E_\mu}\big)\chi
=\sum_{\tilde\nu\in\Lambda^{(n,m)}} \overline{f^{(\tilde\nu)}_0}
\big(\prod_{\substack{\mu\in\Lambda^{(n,m)}\\\mu\neq\nu}}\tfrac{H-E_\mu}{E_\nu-E_\mu}\big)f^{(\tilde\nu)} \\
&=\sum_{\tilde\nu\in\Lambda^{(n,m)}} \overline{f^{(\tilde\nu)}_0}
\big(\prod_{\substack{\mu\in\Lambda^{(n,m)}\\\mu\neq\nu}}\tfrac{E_{\tilde\nu}-E_\mu}{E_\nu-E_\mu}\big)f^{(\tilde\nu)}
=\overline{f^{(\nu)}_0} f^{(\nu)}
\end{align*}
as an identity between real-analytic functions of $p\in (-1,1)$. (Notice in this connection that, as a function of $p\in (-1,1)$, the eigenfunction 
$h^{(\nu)}$ \eqref{h-nu} is real-valued and meromorphic while $\overline{f^{(\nu)}_0} f^{(\nu)}$ is continuous.)
From the equality in question (and the orthonormality of the basis $f^{(\nu)}$, $\nu\in\Lambda^{(n,m)}$) it is clear that $\langle h^{(\nu)}, h^{(\nu)}\rangle_\Delta=| f^{(\nu)}_0 |^2=h^{(\nu)}_0$.

\emph{(iii)} This assertion is immediate from part \emph{(ii)} and the observation that 
\begin{equation*}
 \prod_{\nu\in\Lambda^{(n,m)}} h^{(\nu)}_0\Bigr|_{p=0} >0
\end{equation*}
(by part \emph{(iv)}).

\emph{(iv)}  The operators $H\vert_{p=0}$ and $D^{\mathrm{qR}}$ \cite[Eq. (6.15)]{die-sto:multivariable} differ only by a constant (cf. the proof of Proposition \ref{proposition:simple-spectrum}), and thus have the same eigenfunctions. 
Since the eigenvalues at $p=0$ in Proposition \ref{proposition:simple-spectrum} are nondegenerate as analytic
functions of the coupling parameters, the corresponding eigenfunctions are unique up to normalization.
The asserted formulas for the eigenfunctions in terms of Koornwinder-Macdonald polynomials
go back to \cite[Theorem 6.5]{die-sto:multivariable}. Notice in this connection that these explicit formulas confirm that our normalization is such that
$h_0^{(\nu)}\vert_{p=0}=1/\textsc{n}_{\nu,q}= \langle h^{(\nu)}, h^{(\nu)}\rangle_\Delta \vert_{p=0}$ ($>0$).
\end{proof}

\begin{subequations}
It is immediate from Theorem \ref{theorem:diagonalization} that the eigenfunctions satisfy the following dual orthogonality relation
\begin{equation}\label{dual-orthogonality}
\forall\lambda,\mu\in\Lambda^{(n,m)}:\quad\sum_{\nu\in\Lambda^{(n,m)}} f^{(\nu)}_\lambda \overline{f^{(\nu)}_\mu}  =
\begin{cases}\tfrac{1}{ \Delta_\lambda}&\text{if}\ \lambda=\mu,\\0&\text{if}\ \lambda\neq\mu.\end{cases} 
\end{equation}
In view of Eq. \eqref{eigenfunctions}, this means that

\begin{equation}\label{dual-orthogonality2}
\forall\lambda,\mu\in\Lambda^{(n,m)}:\quad\sum_{\nu\in\Lambda^{(n,m)}}   \frac{h^{(\nu)}_\lambda h^{(\nu)}_\mu}{\langle h^{(\nu)}, h^{(\nu)}\rangle_\Delta} =
\begin{cases}\tfrac{1}{ \Delta_\lambda}&\text{if}\ \lambda=\mu,\\0&\text{if}\ \lambda\neq\mu,\end{cases} 
\end{equation}
unless $p$ belongs to the finite zero-locus $\{ 0<| p|\leq \epsilon \mid \prod_{\nu\in\Lambda^{(n,m)}} h^{(\nu)}_0 =0\}$ for some $\epsilon\in (0,1)$.
\end{subequations}

\begin{note}
From \cite[Theorem 6.5]{die-sto:multivariable} one extracts the following
explicit product formula for the normalization constant $\textsc{n}_{\nu ,q}$ \eqref{snorms:p=0}:
\begin{equation*}
\begin{split}
\textsc{n}_{\nu,q}=&\prod_{1\leq j\leq n}\tfrac{[\mathrm{g}_2-\rho_j,\hat\rho_j-\hat{\mathrm{g}}_2]_{2,q,m}}{[\rho_j-\mathrm{g}'_2+\frac{1}{2},\hat\rho_j-\hat{\mathrm{g}}'_2+\frac{1}{2}]_{2,q,m}}\prod_{r=1,2}\tfrac{[\hat{\rho}_j+\hat{\mathrm{g}}_r,\hat{\rho}_j+1-\hat{\mathrm{g}}_r,\hat{\rho}_j+\hat{\mathrm{g}}'_r+\frac{1}{2},\hat{\rho}_j-\hat{\mathrm{g}}'_r+\frac{1}{2}]_{r,q,\nu_j}}{[\hat{\rho}_j,\hat{\rho}_j+1,\hat{\rho}_j+\frac{1}{2},\hat{\rho}_j+\frac{1}{2}]_{r,q,\nu_j}}\\
&\times\prod_{\substack{1\leq j<k\leq n\\\delta=\pm 1}}\tfrac{[\hat{\rho}_j+\delta\hat{\rho}_k+\mathrm{g},\hat{\rho}_j+\delta\hat{\rho}_k+1-\mathrm{g}]_{1,q,\nu_j+\delta\nu_k}}{[\hat{\rho}_j+\delta\hat{\rho}_k,\hat{\rho}_j+\delta\hat{\rho}_k+1]_{1,q,\nu_j+\delta\nu_k}} .
\end{split}
\end{equation*}
\end{note}

\section{Explicit eigenfunctions in special cases}
\label{sec:5}

According to Theorem \ref{theorem:diagonalization}, the eigenfunctions at the trigonometric level ($p=0$) are given by the multivariable $q$-Racah polynomials of \cite{die-sto:multivariable}. Below we highlight two other special instances for which the eigenfunctions and their normalization can be computed explicitly in terms of polynomials on the spectrum.

\subsection{The case \emph{m}=1: one-column partitions}
\label{subsec:5.1}

Throughout this subsection, we will fix the level at $m=1$. The corresponding set of bounded partitions $\Lambda^{(n,1)}$  consists of columns of the form
\begin{equation}\label{columns}
(1^k)=(\underbrace{1,\dots,1}_{k},\underbrace{0,\dots,0}_{n-k}),\quad 0\leq k\leq n
\end{equation}
(cf. Eq. \eqref{bounded-partitions}). Upon extending $\rho_j=(n-j)\mathrm{g}+\mathrm{g}_1$ \eqref{rho} to the indices $j=0$ and $j=n+1$ for convenience, we  list the non-vanishing coefficients of $H$ \eqref{Hamiltonian}--\eqref{C-r} in this situation:
\begin{subequations}
\begin{align}\label{A-k}
A_{(1^k)}=& \sum_{1\leq r\leq 4}\mathrm{c}_r\Bigl( 
\prod_{1\leq j\leq k} 
\tfrac{[\rho_j+\frac{3}{2}-\mathrm{g},\rho_j+\frac{1}{2}+\mathrm{g}]_r}{[\rho_j+\frac{3}{2},\rho_j+\frac{1}{2}]_r}
\prod_{k<j\leq n}
\tfrac{[\rho_j+\frac{1}{2}-\mathrm{g},\rho_j-\frac{1}{2}+\mathrm{g}]_r}{[\rho_j+\frac{1}{2},\rho_j-\frac{1}{2}]_r}
-1\Bigr) \\
=& \sum_{1\leq r\leq 4}\mathrm{c}_r\Bigl( \tfrac{[\rho_0+\frac{1}{2},\rho_{k+1}+\frac{3}{2},\rho_k-\frac{1}{2},\rho_{n+1}+\frac{1}{2}]_r}{[\rho_k+\frac{1}{2},\rho_1+\frac{3}{2},\rho_n-\frac{1}{2},\rho_{k+1}+\frac{1}{2}]_r}-1\Bigr) \nonumber
\end{align}
with $0\leq k \leq n$,
\begin{align}\label{B-k+1}
B_{(1^k),k+1}=& \prod_{1\leq r\leq 4}\tfrac{[\rho_{k+1}+\mathrm{g}_r,\rho_{k+1}+\mathrm{g}'_r+\frac{1}{2}]_r}{[\rho_{k+1},\rho_{k+1}+\frac{1}{2}]_r}
\prod_{\substack{1\leq j\leq k\\\delta=\pm 1}} \tfrac{[\rho_{k+1}+\delta\rho_j+\delta+\mathrm{g}]_1 }{[\rho_{k+1}+\delta\rho_j+\delta]_1} 
\prod_{\substack{k+1<j\leq n\\ \delta=\pm 1}} \tfrac{[\rho_{k+1}+\delta\rho_j+\mathrm{g}]_1 }{[\rho_{k+1}+\delta\rho_j]_1}  
\\
=&\tfrac{[\rho_0+\rho_{k+1}+1,2\rho_{k+1},\rho_{k+1}-\rho_{n+1},1]_1}{[\rho_k+\rho_{k+1}+1,\rho_1-\rho_{k+1}+1,\rho_{k+1}+\rho_n,\mathrm{g}]_1}\prod_{1\leq r\leq 4}\tfrac{[\rho_{k+1}+\mathrm{g}_r,\rho_{k+1}+\mathrm{g}'_r+\frac{1}{2}]_r}{[\rho_{k+1},\rho_{k+1}+\frac{1}{2}]_r} \nonumber
\end{align}
with $0\leq k < n$, and
\begin{align}\label{B-k}
B_{(1^k),-k}=&
\prod_{1\leq r\leq 4}\tfrac{[\rho_k+1-\mathrm{g}_r,\rho_k-\mathrm{g}'_r+\frac{1}{2}]_r}{[\rho_k+1,\rho_k+\frac{1}{2}]_r} 
\prod_{\substack{1\leq j< k\\\delta=\pm 1}} 
\tfrac{[\rho_{k}+\delta\rho_j+1+\delta-\mathrm{g}]_1 }{[\rho_{k}+\delta\rho_j+1+\delta]_1} 
\prod_{\substack{k<j\leq n\\ \delta=\pm 1}} \tfrac{[\rho_{k}+\delta\rho_j+1-\mathrm{g}]_1 }{[\rho_{k}+\delta\rho_j+1]_1}
\\
=&\tfrac{[2\rho_{k+1}+2,\rho_0-\rho_k,\rho_k+\rho_{n+1}+1,1]_1}{[\rho_1+\rho_k+2,\rho_k+\rho_{k+1}+1,\rho_k-\rho_n+1,\mathrm{g}]_1}
\prod_{1\leq r\leq 4}\tfrac{[\rho_k+1-\mathrm{g}_r,\rho_k-\mathrm{g}'_r+\frac{1}{2}]_r}{[\rho_k+1,\rho_k+\frac{1}{2}]_r} \nonumber
\end{align}
\end{subequations}
with $0< k \leq n$.  The eigenvalue problem for $H$ \eqref{Hamiltonian2} in
$\mathcal{C}(\Lambda^{(n,1)})$ thus takes the following tridiagonal form:
\begin{equation}\label{recrel}
B_{-k}f_{(1^{k-1})}+A_k f_{(1^k)} +B_{k+1}f_{(1^{k+1})}=Ef_{(1^k)} ,\quad 0\leq k\leq n ,
\end{equation}
where we have employed the short-hand notation $A_k=A_{(1^k)}$ \eqref{A-k}, $B_{k+1}=B_{(1^k),k+1}$  \eqref{B-k+1}, $B_{-k}=B_{(1^k),-k}$  \eqref{B-k} with the convention that $B_0=B_{n+1}=0$.
With the aid of the recurrence in Lemma \ref{lemma:recurrence}, we can now express the corresponding weight function $\Delta\in\mathcal{C}(\Lambda^{(n,1)})$ \eqref{weights} in terms of the coefficients as follows:
\begin{align}\label{weights3}
\Delta_{(1^k)}=&  \prod_{1\leq j\leq k}\tfrac{[2\rho_j+2]_1}{[2\rho_j]_1}\prod_{1\leq r\leq 4}\tfrac{[\rho_j+\mathrm{g}_r,\rho_j+\mathrm{g}'_r+\frac{1}{2}]_{r}}{[\rho_j+1-\mathrm{g}_r,\rho_j-\mathrm{g}'_r+\frac{1}{2}]_{r}}\\
&\times \prod_{1\leq j<l\leq k}\tfrac{[\rho_j+\rho_l+2]_1}{[\rho_j+\rho_l]_1}
\tfrac{[\rho_j+\rho_l+\mathrm{g}]_{1}[\rho_j+\rho_l+1+\mathrm{g}]_{1}}{[\rho_j+\rho_l+1-\mathrm{g}]_{1}[\rho_j+\rho_l+2-\mathrm{g}]_{1}}   \nonumber \\
&\times \prod_{\substack{1\leq j\leq k \\ k<l\leq n}}\tfrac{[\rho_j+\rho_l+1]_1}{[\rho_j+\rho_l]_1}
\tfrac{[\rho_j+\rho_l+\mathrm{g}]_{1}}{[\rho_j+\rho_l+1-\mathrm{g}]_{1}}
\tfrac{[\rho_j-\rho_l+1]_1}{[\rho_j-\rho_l]_1}
\tfrac{[\rho_j-\rho_l+\mathrm{g}]_{1}}{[\rho_j-\rho_l+1-\mathrm{g}]_{1}}   \nonumber \\
=&\prod_{1\leq j\leq k}B_jB_{-j}^{-1},\quad 0\leq k\leq n  \nonumber
\end{align}
(so $\Delta_{(1^0)}=1$).

To construct the eigenfunctions let us define $P_{(1^0)}(E)=1$ and
\begin{subequations}
\begin{equation}\label{level1polynomials}
P_{(1^k)}(E)=\det\begin{bmatrix}
E-A_0&-B_1&0&\cdots&0\\
-B_{-1}&E-A_1&\ddots&&\vdots\\
0&-B_{-2}&\ddots&-B_{k-2}&0\\
\vdots&&\ddots&E-A_{k-2}&-B_{k-1}\\
0&\cdots&0&-B_{-k+1}&E-A_{k-1}
\end{bmatrix}
\end{equation}
for $k=1,\dots,n+1$ (with $(1^{n+1})\in \Lambda^{(n+1,1)}$). Notice that $P_{(1^k)}(E)$ is a monic polynomial of degree $k$ in $E$. Upon expanding the determinant for $P_{(1^{k+1})}(E)$ along its last row/column, it is seen that the polynomials in question obey the three-term recurrence relation
\begin{equation}\label{recrel3}
P_{(1^{k+1})}(E)=(E-A_k)P_{(1^{k})}(E)-B_kB_{-k}P_{(1^{k-1})}(E),\quad 0\leq k\leq n.
\end{equation}
\end{subequations}

\begin{theorem}[Diagonalization for \emph{m}=1]\label{theorem:level1}
The following statements hold for parameters subject to the conditions in Eqs. \eqref{coupling-conditions}, \eqref{truncation} with $m=1$ and $\mathrm{g}\neq 1$.

\begin{subequations}
\emph{(i)} The eigenvalues from Proposition \ref{proposition:simple-spectrum} are given by the simple roots
\begin{equation}
E_{(1^0)}>E_{(1^1)}> \cdots > E_{(1^n)}
\end{equation}
of $P_{(1^{n+1})}(E)$ \eqref{level1polynomials}.

\emph{(ii)} The eigenfunctions $h^{(1^l)}$, $0\leq l\leq n$ from Theorem \ref{theorem:diagonalization} are given by
\begin{equation}
 h_{(1^k)}^{(1^l)}= \frac{\textsc{c}_{(1^k)}}{ \textsc{n}_{(1^l)} }  P_{(1^k)}(E_{(1^l)})\qquad (0\leq k,l\leq n)
\end{equation}
with
\begin{align}
\textsc{c}_{(1^k)}= & \prod_{\substack{1\leq j\leq k\\1\leq r\leq 4}}\tfrac{[\rho_j,\rho_j+\frac{1}{2}]_{r}}{[\rho_j+\mathrm{g}_r,\rho_j+\mathrm{g}'_r+\frac{1}{2}]_{r}}
\prod_{1\leq j<l\leq k}\tfrac{[\rho_j+\rho_l]_{1}[\rho_j+\rho_l+1]_{1}}{[\rho_j+\rho_l+\mathrm{g}]_{1}[\rho_j+\rho_l+1+\mathrm{g}]_{1}}  \\
&\times \prod_{\substack{1\leq j\leq k\\ k<l\leq n}}  \tfrac{[\rho_j+\rho_l]_{1}[\rho_j-\rho_l]_{1}}{[\rho_j+\rho_l+\mathrm{g}]_{1}[\rho_j-\rho_l+\mathrm{g}]_{1}} \nonumber \\
=& \prod_{1\leq j\leq k}B_j^{-1} \nonumber
\end{align}
(so $\textsc{c}_{(1^0)}=1$) and
\begin{align}
\textsc{n}_{(1^l)} &= \sum_{0\leq k\leq n}   \textsc{c}_{(1^k)}^2  P_{(1^k)}^2(E_{(1^l)}) \Delta_{(1^k)} \\
&= \textsc{c}_{(1^n)}^2\Delta_{(1^n)} P_{(1^n)}\bigl( E_{(1^l)}\bigr)
\prod_{\substack{0\leq j\leq n\\j\neq l}}\left(E_{(1^l)}-E_{(1^j)}\right) .\nonumber
\end{align}
\end{subequations}
\end{theorem}

\begin{proof}

\emph{(i)} It is immediate from the eigenvalue equation in Eq. \eqref{recrel} that
$P_{(1^{n+1})}(E)$ \eqref{level1polynomials} produces the characteristic polynomial of $H$ \eqref{Hamiltonian2} in the present situation.
The corresponding eigenvalues are thus given by the roots of $P_{(1^{n+1})}(E)$. Since the off-diagonal elements of the underlying real tridiagonal matrix are positive by Lemma \ref{lemma:truncation}, the eigenvalues in question are real and simple 
(cf. e.g.  \cite[Chapter III.11.4]{pra:problems}). From the analyticity  in $-1<p<1$ and the explicit eigenvalues at $p=0$ given by
\begin{equation*}
E_{(1^l)}\vert_{p=0}=
2\sum_{1\leq j\leq l} {\textstyle \cos\Bigl( \frac{\pi (1+\hat\rho_j)}{1+(n-1)\mathrm{g}+\mathrm{g}_1+\mathrm{g}_2}\Bigr)}
+ 2\sum_{l< j\leq n} {\textstyle \cos\Bigl( \frac{\pi \hat\rho_j}{1+(n-1)\mathrm{g}+\mathrm{g}_1+\mathrm{g}_2}\Bigr)} 
\end{equation*}
(cf. Eq. \eqref{trig-eigenvalues}),
it is clear that Proposition \ref{proposition:simple-spectrum} labels the simple eigenvalues for $m=1$ in the decreasing order as indicated (since  level-crossings cannot occur when deforming the multiplicity-free spectrum away from $p=0$).

\emph{(ii)} It follows from the recurrence in Eq. \eqref{recrel3} that $f_k= \textsc{c}_{(1^k)}  P_{(1^k)}(E)$
solves the eigenvalue equation \eqref{recrel} provided $E$ is a root of $P_{(1^{n+1})}(E)$. Since $f_0=1$, this yields
a nontrivial eigenfunction for each eigenvalue. In order to normalize in accordance with Theorem \ref{theorem:diagonalization},
it suffices to divide the eigenfunction for $E_{(1^l)}$ by
\begin{equation*}
\textsc{n}_{(1^l)}= \sum_{0\leq k\leq n}   \textsc{c}_{(1^k)}^2  P_{(1^k)}^2(E_{(1^l)}) \Delta_{(1^k)} .
\end{equation*}
By invoking the Christoffel-Darboux formula (cf. the note below for some additional details), we find that
\begin{equation*}
   \textsc{n}_{(1^l)} =  \sum_{0\leq k\leq n} 
   \tfrac{P_{(1^k)}^2(E_{(1^l)})}{\prod_{1\leq j\leq k}B_jB_{-j}}
=\tfrac{P_{(1^n)}(E_{(1^l)})}{\prod_{1\leq j\leq n}B_jB_{-j}}\prod_{\substack{0\leq j\leq n\\j\neq l}}(E_{(1^l)}-E_{(1^j)})
\end{equation*}
as stated.
\end{proof}

\begin{note} The polynomials $P_{(1^k)}(E)$ \eqref{level1polynomials} satisfy
the Christoffel-Darboux identity \cite{sze:orthogonal}
\begin{equation*}
\sum_{0\leq k\leq n}\tfrac{P_{(1^k)}(x)P_{(1^k)} (y)}{\prod_{1\leq j\leq k}B_jB_{-j}}=\tfrac{P_{(1^{n+1})}(x) P_{(1^n)}(y)-P_{(1^n)}(x)P_{(1^{n+1})}(y)}{(x-y)\prod_{1\leq j\leq n}B_jB_{-j}}
\end{equation*}
and its $y\to x$ confluent limit
\begin{equation*}
\sum_{0\leq k\leq n}\tfrac{P_{(1^k)}^2(x)}{\prod_{1\leq j\leq k}B_jB_{-j}}=\tfrac{P'_{(1^{n+1})}(x)P_{(1^n)}(x)-P'_{(1^n)}(x)P_{(1^{n+1})}(x)}{\prod_{1\leq j\leq n}B_jB_{-j}}.
\end{equation*}
Indeed, the three-term recurrence relation in Eq. \eqref{recrel3} establishes that
\begin{multline*}
P_{(1^{n+1})}(x)P_{(1^n)}(y)-P_{(1^n)}(x)P_{(1^{n+1})}(y)=\\(x-y)P_{(1^n)}(x)P_{(1^n)}(y)
+B_nB_{-n}(P_{(1^n)}(x)P_{(1^{n-1})}(y)-P_{(1^{n-1})}(x)P_{(1^n)}(y)) ,
\end{multline*}
which recovers the Christoffel-Darboux identity by iteration:
\begin{align*}
P_{(1^{n+1})}(x)P_{(1^n)}(y) & -P_{(1^n)}(x)P_{(1^{n+1})}(y)= \\
&(x-y)\sum_{0\leq k\leq n}\Big(P_{(1^k)}(x)P_{(1^k)}(y)\prod_{k< j\leq n}B_jB_{-j}\Big)
\end{align*}
(upon dividing both sides by $(x-y)\prod_{1\leq j\leq n}B_jB_{-j}$).  The computation of the quadratic norms
in (the last formula of) the proof
of Theorem \ref{theorem:level1} hinges on the confluent Christoffel-Darboux identity evaluated at $x=E_{(1^l)}$,
in combination with the factorization of the characteristic polynomial over the eigenvalues:
\begin{equation*}
P_{(1^{n+1})}(x)=(x-E_{(1^0)})(x-E_{(1^1)})\cdots (x-E_{(1^n)}) .
\end{equation*}
\end{note}

\subsection{The case g=1: generalized Schur polynomials}
\label{subsec:5.2}

In the limit $\mathrm{g}\to 1$, the coefficient $A_\lambda$ \eqref{A-lambda} remains regular:
\begin{subequations}
\begin{equation}
A_\lambda \vert_{\mathrm{g}\to 1}= 
\sum_{1\leq r\leq 4} \text{Res}_{\mathrm{g}=1}(\mathrm{c}_r)\sum_{1\leq j\leq n}\bigg(\tfrac{[\rho_j+\lambda_j+\frac{1}{2}]'_r}{[\rho_j+\lambda_j+\frac{1}{2}]_r}-\tfrac{[\rho_j+\lambda_j-\frac{1}{2}]'_r}{[\rho_j+\lambda_j-\frac{1}{2}]_r}\bigg),
\end{equation}
with
\begin{equation}
 \text{Res}_{\mathrm{g}=1}(\mathrm{c}_r)=\tfrac{2}{[1]_1}\prod_{1\leq s\leq 4}[\mathrm{g}_{\pi_r(s)}-\tfrac{1}{2}]_s[\mathrm{g}'_{\pi_r(s)}]_s.
\end{equation}
\end{subequations}
Here $[z]'_r$ denotes the derivative of the scaled theta function $[z]_r$  \eqref{thetas} with respect to the variable $z$. Moreover, the coefficients $B_{\lambda,\varepsilon j}$ \eqref{B-lambda}
and the orthogonality weights  $\Delta_\lambda$ \eqref{Delta2} simplify at this value of the coupling parameter $\mathrm{g}$:
\begin{subequations}
\begin{equation}
B_{\lambda,\varepsilon j} \vert_{\mathrm{g}\to 1}
=
\tfrac{V_{\lambda+\varepsilon e_j}}{V_\lambda}
\prod_{1\leq r\leq 4}\tfrac{[\rho_j+\lambda_j+\varepsilon\mathrm{g}_r]_r[\rho_j+\lambda_j+\varepsilon(\mathrm{g}'_r+\frac{1}{2})]_r}{[\rho_j+\lambda_j]_r[\rho_j+\lambda_j+\frac{\varepsilon}{2}]_r}
\end{equation}
and
\begin{equation}
\Delta_\lambda  \vert_{\mathrm{g}\to 1}=  V_\lambda^2 \prod_{1\leq j\leq n}
\tfrac{[2\rho_j+2\lambda_j]_1}{[2\rho_j]_1}\prod_{1\leq r\leq 4}\tfrac{[\rho_j+\mathrm{g}_r,\rho_j+\mathrm{g}'_r+\frac{1}{2}]_{r,\lambda_j}}{[\rho_j+1-\mathrm{g}_r,\rho_j-\mathrm{g}'_r+\frac{1}{2}]_{r,\lambda_j}} 
\end{equation}
with
\begin{equation}
V_\lambda=\prod_{\substack{1\leq j<k\leq n\\\delta=\pm 1}}\tfrac{[\rho_j+\delta\rho_k+\lambda_j+\delta\lambda_k]_1}{[\rho_j+\delta\rho_k]_1} .
\end{equation}
\end{subequations}

We learn from these formulas that
\begin{subequations}
\begin{equation}\label{A-lambda-g=1:a}
A_\lambda  \vert_{\mathrm{g}\to 1} = \sum_{1\leq j\leq n}a_{n-j+\lambda_j}
\end{equation}
with
\begin{equation}\label{A-lambda-g=1:b}
a_k=\sum_{1\leq r\leq 4} \text{Res}_{\mathrm{g}=1}(\mathrm{c}_r)
\bigg(\tfrac{[\mathrm{g}_1+k+\frac{1}{2}]'_r}{[\mathrm{g}_1+k+\frac{1}{2}]_r}-\tfrac{[\mathrm{g}_1+k-\frac{1}{2}]'_r}{[\mathrm{g}_1+k-\frac{1}{2}]_r}\bigg) ,
\end{equation}
and
\begin{equation}\label{B-lambda-g=1:a}
 B_{\lambda,j} \vert_{\mathrm{g}\to 1}=\tfrac{V_{\lambda+ e_j}}{V_\lambda}  b_{n-j+\lambda_j}^{+},\quad
B_{\lambda,-j} \vert_{\mathrm{g}\to 1}=\tfrac{V_{\lambda -e_j}}{V_\lambda} b_{n-j+\lambda_j}^{-}
\end{equation}
with
\begin{equation}\label{B-lambda-g=1:b}
b_{k}^{+}=\prod_{1\leq r\leq 4}\tfrac{[\mathrm{g}_1+\mathrm{g}_r+k]_r[\mathrm{g}_1+\mathrm{g}'_r+\frac{1}{2}+k]_r}{[\mathrm{g}_1+k]_r[\mathrm{g}_1+\frac{1}{2}+k]_r},\quad
b_{k}^{-}=\prod_{1\leq r\leq 4}\tfrac{[\mathrm{g}_1-\mathrm{g}_r+k]_r[\mathrm{g}_1-\mathrm{g}'_r-\frac{1}{2}+k]_r}{[\mathrm{g}_1+k]_r[\mathrm{g}_1-\frac{1}{2}+k]_r} 
\end{equation}
(so $b_0^-=b_{n+m-1}^+=0$). 
\end{subequations}
Moreover, upon indicating the dependence on $n$ and $m$ explicitly, the elliptic weights can be factorized correspondingly:
\begin{subequations}
\begin{equation}\label{weights:g=1}
\Delta^{(n,m)}_\lambda  \vert_{\mathrm{g}\to 1}=  V_\lambda^2 \prod_{1\leq j\leq n} \tfrac{\Delta^{(1,n+m-1)}_{n-j+\lambda_j}}{\Delta^{(1,n+m-1)}_{n-j}}
\end{equation}
with
\begin{align}
\Delta^{(1,n+m-1)}_k =& \tfrac{[2\mathrm{g}_1+2k]_1}{[2\mathrm{g}_1]_1}\prod_{1\leq r\leq 4}\tfrac{[\mathrm{g}_1+\mathrm{g}_r,\mathrm{g}_1+\mathrm{g}'_r+\frac{1}{2}]_{r,k}}{[\mathrm{g}_1+1-\mathrm{g}_r,\mathrm{g}_1-\mathrm{g}'_r+\frac{1}{2}]_{r,k}} \quad (0\leq k <n+m), \\
=&   \prod_{0\leq j< k} b_j^+ (b_{j+1}^-)^{-1}   . \nonumber
\end{align}
\end{subequations}

For $0\leq k\leq n+m$, let us now define the following monic polynomials
\begin{subequations}
\begin{equation}\label{E-Racah}
P_{k}(E)=\det\begin{bmatrix}
E-a_0&-{b}_{0}^+&0&\cdots&0\\
-b_{1}^-&E-a_1&\ddots&&\vdots\\
0&-b_{2}^-&\ddots&-{b}_{k-3}^+&0\\
\vdots&&\ddots&E-a_{k-2}&-{b}_{k-2}^+\\
0&\cdots&0&-b_{k-1}^-&E-a_{k-1}
\end{bmatrix}
\end{equation}
(so $P_{0}(E)=1$), which obey the three-term recurrence
\begin{equation}\label{E-Racah-rec}
P_{k+1}(E)=(E-a_k)P_{k}(E)-b_{k-1}^+ b_{k}^- P_{k-1}(E),\quad 0\leq k < n+m .
\end{equation}
\end{subequations}
The polynomials in question amount to the elliptic Racah basis of \cite{die-gor:racah} on  $n+m$ nodes. 
From \cite[Proposition 7]{die-gor:racah}, one extracts the following (dual) orthogonality relation for
$P_0(E),P_1(E),\ldots ,P_{n+m-1}(E)$:

\begin{subequations}
\begin{equation}\label{orthogonality2}
\sum_{0\leq k<n+m} \textsc{c}_k^2 P_k(E_j)P_k(E_l)\Delta^{(n+m-1)}_k=\begin{cases}\textsc{n}_l&\text{if}\ j=l,\\0&\text{if}\ j\neq l,\end{cases}
\end{equation}
where
\begin{equation}\label{E-racah:roots}
E_0>E_1>\cdots > E_{n+m-1}
\end{equation}
denote the roots of the top-degree elliptic Racah polynomial $P_{n+m}(E)$, and the normalizations are governed by
\begin{equation}\label{ck}
\textsc{c}_k=\prod_{0\leq j <k}  (b_j^+)^{-1} =
 \prod_{ 1\leq r\leq 4}   \tfrac{[\mathrm{g}_1,\mathrm{g}_1+\frac{1}{2}]_{r,k}}{[\mathrm{g}_1+\mathrm{g}_r,\mathrm{g}_1+\mathrm{g}'_r+\frac{1}{2}]_{r,k}}
\end{equation}
and
\begin{align}
\textsc{n}_l=&\frac{P_{n+m-1}( E_l)}{\prod_{1\leq k<n+m}b_{k-1}^+ b_k^-}\prod_{\substack{0\leq j<n+m\\j\neq l}} ( E_l-E_j )  \\
=&\textsc{c}_{n+m-1}^2 \Delta^{(1,n+m-1)}_{n+m-1} P_{n+m-1}(E_l)  \prod_{\substack{0\leq j<n+m\\j\neq l}} ( E_l-E_j )  . \nonumber
\end{align}
\end{subequations}
Notice that it is straightforward to check these orthogonality relations from scratch via the Christoffel-Darboux formula by adapting the corresponding computations in Subsection \ref{subsec:5.1}.

\begin{theorem}[Diagonalization for g=1]\label{theorem:g=1}
The following statements hold for parameters subject to the conditions in Eqs. \eqref{coupling-conditions}, \eqref{truncation}.

\begin{subequations}
\emph{(i)}  For $\mathrm{g}\to 1$, the eigenvalues from Proposition \ref{proposition:simple-spectrum} can be written explicitly in terms of the roots $E_l$ \eqref{E-racah:roots}
of the top-degree elliptic Racah polynomial $P_{n+m}(E)$ \eqref{E-Racah}, \eqref{E-Racah-rec}:
\begin{equation}\label{eigenvalues:g=1}
E_\nu=\sum_{1\leq j\leq n} E_{n-j+\nu_j}\qquad (\nu\in\Lambda^{(n,m)}) .
\end{equation}

\emph{(ii)} For $\mathrm{g}\to 1$ an orthogonal basis of eigenfunctions $h^{(\nu )}$, $\nu\in\Lambda^{(n,m)}$ such that $H h^{(\nu )} =E_\nu h^{(\nu )}$ and
$h^{(\nu )}_0= \langle h^{(\nu )}, h^{(\nu )}\rangle_\Delta$, can be written down explicitly in terms of the
generalized Schur polynomials associated with the elliptic Racah basis:
\begin{equation}\label{eigenfunctions:g=1}
 h_{\lambda }^{(\nu)}= \tfrac{\textsc{c}_{\lambda}}{ \textsc{n}_{\nu} }  \textsc{s}^{(\nu)}_\lambda \qquad (\lambda ,\nu\in\Lambda^{(n,m)}) 
\end{equation}
with
\begin{equation}\label{schur}
\textsc{s}^{(\nu)}_\lambda=\textsc{a}^{(\nu)}_\lambda / \textsc{a}^{(\nu)}_0, \quad \textsc{a}^{(\nu)}_\lambda=\det\big[P_{n-i+\lambda_i}(E_{n-j+\nu_j})\big]_{1\leq i,j\leq n} .
\end{equation}
Here the normalizations are governed by
\begin{equation}\label{C-lambda:g=1}
\textsc{c}_{\lambda}=\tfrac{1}{V_\lambda} \prod_{\substack{1\leq j\leq n\\ 1\leq r\leq 4}}\tfrac{[\rho_j,\rho_j+\frac{1}{2}]_{r,\lambda_j}}{[\rho_j+\mathrm{g}_r,\rho_j+\mathrm{g}'_r+\frac{1}{2}]_{r,\lambda_j}} 
\end{equation}
and
\begin{equation}\label{N:g=1}
\textsc{n}_{\nu } =\sum_{\lambda\in\Lambda^{(n,m)}}
\textsc{c}_\lambda^2 (\textsc{s}^{(\nu)}_\lambda)^2  \Delta_\lambda  =
\frac{1}{ \bigl( \textsc{a}^{(\nu)}_0\bigr)^2}  \prod_{1\leq j\leq n} \frac{\textsc{n}_{n-j+\nu_j}}{\Delta^{(1,n+m-1)}_{n-j} \textsc{c}^2_{n-j} } . 
\end{equation}
\end{subequations}
\end{theorem}

\begin{proof}
For any $\mu\in\Lambda^{(n)}$ let $\mu(i)\in\Lambda^{(n-1)}$, $i=1,\dots,n$ denote the partition obtained from $\mu$ by omitting its $i$-th part, i.e.
$\mu(1)=(\mu_2,\dots,\mu_n)$, $\mu(n)=(\mu_1,\dots,\mu_{n-1})$, and
$\mu(j)=(\mu_1,\dots,\mu_{j-1},\mu_{j+1},\dots,\mu_n)$ for $1<j<n$.
We then compute the action of $H\vert_{\mathrm{g}\to 1}$  on the function $ \textsc{c}_\lambda \textsc{s}^{(\nu)}_\lambda$ as follows:
\begin{align*}
&(H\vert_{\mathrm{g}\to 1} \textsc{c} \textsc{s}^{(\nu)})_\lambda\overset{\eqref{Hamiltonian2}}{=}A_\lambda\vert_{\mathrm{g}\to 1}\textsc{c}_\lambda \textsc{s}^{(\nu)}_\lambda+\sum_{\substack{1\leq i\leq n,\,\varepsilon=\pm 1\\\lambda+\varepsilon e_i\in\Lambda^{(n,m)}}}B_{\lambda,\varepsilon i} \vert_{\mathrm{g}\to 1}
\textsc{c}_{\lambda+\varepsilon e_i} \textsc{s}^{(\nu)}_{\lambda+\varepsilon e_i}\\
&\overset{\eqref{A-lambda-g=1:a}-\eqref{B-lambda-g=1:b}}{=}\\
&\textsc{c}_\lambda
\sum_{1\leq i\leq n}(a_{n-i+\lambda_i}\textsc{s}^{(\nu)}_\lambda +\tfrac{\textsc{c}_{\lambda -e_i}}{\textsc{c}_\lambda} \tfrac{V_{\lambda -e_i}}{V_\lambda} b^-_{n-i+\lambda_i}\textsc{s}^{(\nu)}_{\lambda-e_i}+\tfrac{\textsc{c}_{\lambda +e_i}}{\textsc{c}_\lambda}
\tfrac{V_{\lambda +e_i}}{V_\lambda}
b_{n-i+\lambda_i}^+\textsc{s}^{(\nu)}_{\lambda+e_i})
\end{align*}

\begin{equation*}
\begin{split}
&\overset{\eqref{C-lambda:g=1}}{=}
\textsc{c}_\lambda
\sum_{1\leq i\leq n}(a_{n-i+\lambda_i}\textsc{s}^{(\nu)}_\lambda +b^+_{n-i+\lambda_i-1} b^-_{n-i+\lambda_i}\textsc{s}^{(\nu)}_{\lambda-e_i}+\textsc{s}^{(\nu)}_{\lambda+e_i})\\
&\overset{\eqref{schur}}{=}\frac{\textsc{c}_\lambda}{\textsc{a}^{(\nu)}_0}
\sum_{1\leq i\leq n}(a_{n-i+\lambda_i}\textsc{a}^{(\nu)}_\lambda +b^+_{n-i+\lambda_i-1} b^-_{n-i+\lambda_i}\textsc{a}^{(\nu)}_{\lambda-e_i}+\textsc{a}^{(\nu)}_{\lambda+e_i})\\
&\overset{\eqref{E-Racah-rec}}{=}
\frac{\textsc{c}_\lambda}{\textsc{a}^{(\nu)}_0}
\sum_{1\leq i\leq n}\det\big[(\delta_{ik}E_{n-j+\nu_j}+1-\delta_{ik})P_{n-k+\lambda_k}(E_{n-j+\nu_j})\big]_{1\leq k,j\leq n}\\
&\overset{(\ast)}{=}
\frac{\textsc{c}_\lambda}{\textsc{a}^{(\nu)}_0}
\sum_{1\leq i\leq n}\sum_{1\leq j\leq n}(-1)^{i+j}E_{n-j+\nu_j}P_{n-i+\lambda_i}(E_{n-j+\nu_j})\textsc{a}^{(\nu(j))}_{\lambda(i)}\\
&=
\frac{\textsc{c}_\lambda}{\textsc{a}^{(\nu)}_0}
\sum_{1\leq j\leq n}E_{n-j+\nu_j}\sum_{1\leq i\leq n}(-1)^{i+j}P_{n-i+\lambda_i}(E_{n-j+\nu_j})\textsc{a}^{(\nu(j))}_{\lambda(i)}\\
&\overset{(\ast\ast)}{=}\frac{ \textsc{c}_\lambda}{\textsc{a}^{(\nu)}_0}\sum_{1\leq j\leq n}E_{n-j+\nu_j}\textsc{a}^{(\nu)}_{\lambda}\\
&\overset{\eqref{schur}}{=} \textsc{c}_\lambda \textsc{s}^{(\nu)}_\lambda \sum_{1\leq j\leq n}E_{n-j+\nu_j}.
\end{split}
\end{equation*}
Here $\delta_{ik}$ refers to Kronecker's delta symbol, and
in step ($\ast$) the determinant in the $i$-th term is expanded along its $i$-th row for $1\leq i\leq n$, while in step ($\ast\ast$) $\textsc{a}^{(\nu)}_\lambda$ is obtained from its expansion along the $j$-th column for $1\leq j\leq n$.

The above computation shows that   $h^{(\nu )}$ \eqref{eigenfunctions:g=1} solves the eigenvalue equation for $H\vert_{\mathrm{g}\to 1}$  with eigenvalue $E_\nu$ \eqref{eigenvalues:g=1}.
We will now check that  these solutions produce an orthogonal eigenbasis with the aid of the Cauchy-Binet formula (viewed as a discrete counterpart of a determinantal integration  formula due to
Andr\'eief \cite{for:meet}).
Indeed, for any $\nu,\tilde{\nu}\in\Lambda^{(n,m)}$ one has that
\begin{align*}
& \sum_{\lambda\in\Lambda^{(n,m)}}
\textsc{c}_\lambda^2 \textsc{s}^{(\nu)}_\lambda \textsc{s}^{(\tilde\nu)}_\lambda \Delta_\lambda 
= \\
&\frac{1}{  \textsc{a}^{(\nu)}_0 \textsc{a}^{(\tilde\nu)}_0 \prod_{1\leq j\leq n} \Delta^{(1,n+m-1)}_{n-j} \textsc{c}^2_{n-j} }
\sum_{\lambda\in\Lambda^{(n,m)}}       \textsc{a}^{(\nu)}_\lambda \textsc{a}^{(\tilde\nu)}_\lambda  \prod_{1\leq j\leq n}\Delta^{(1,n+m-1)}_{n-j+\lambda_i} \textsc{c}^2_{n-j+\lambda_j} 
\end{align*}
(cf. Eqs. \eqref{weights:g=1}, \eqref{ck} and \eqref{C-lambda:g=1}).
By writing  $n-j+\lambda_j =\mu_j$ we cast the sum on the second line in the form
\begin{align}\label{sum}
\sum_{m+n>\mu_1>\dots>\mu_n\geq 0} & \Bigl( \det [P_{\mu_i}(E_{n-j+\nu_j}) ( \Delta^{(1,n+m-1)}_{\mu_i} )^{1/2}  \textsc{c}_{\mu_i}  ]_{1\leq i,j\leq n}\\
&\times \det [P_{\mu_i}(E_{n-j+\tilde{\nu}_j}) ( \Delta^{(1,n+m-1)}_{\mu_i} )^{1/2}  \textsc{c}_{\mu_i}  ]_{1\leq i,j\leq n} \Bigr) .
\nonumber
\end{align}
Upon comparing with the Cauchy-Binet formula for the expansion of the determinant of the product of the $n\times (n+m)$ matrix 
\begin{equation*}
[P_{k}(E_{n-j+\nu_j}) ( \Delta^{(1,n+m-1)}_{k} )^{1/2}  \textsc{c}_{k}  ]_{\substack { 1\leq j\leq n \\ 0\leq k < n+m}}
\end{equation*}
 and the
$ (n+m)\times n $ matrix 
\begin{equation*}
 [P_{k}(E_{n-j+\tilde{\nu}_j}) ( \Delta^{(1,n+m-1)}_{k} )^{1/2}  \textsc{c}_{k}  ]_{\substack { 0\leq k < n+m\\1\leq j\leq n }} ,
 \end{equation*}
it is seen that the sum in Eq. \eqref{sum} is equal to following determinant:
\begin{align*}
\det\bigg[ \sum_{0\leq k<n+m}P_k(E_{n-i+\nu_i})P_k(E_{n-j+\tilde{\nu}_j})\Delta^{(1,n+m-1)}_k\textsc{c}_k^2 \bigg]_{1\leq i,j\leq n} &\\
\stackrel{\eqref{orthogonality2}}{=}\begin{cases}
\prod_{1\leq j\leq n}\textsc{n}_{n-j+\nu_j}&\text{if}\ \tilde{\nu}=\nu ,\\
0&\text{otherwise}. 
\end{cases} &
\end{align*}
Notice at this point that if
$l=\max\{ 1\leq l\leq n \mid \nu_j = \tilde{\nu}_j\, \text{for all}\, 1\leq j\leq l \}$, then the orthogonality
\eqref{orthogonality2} implies that
the latter determinant factorizes as a product of 
$\prod_{1\leq j\leq l}\textsc{n}_{n-j+\nu_j}$ (stemming from the upper-left $l\times l$ diagonal principal minor) and a factor given by the bottom-right $(n-l)\times (n-l)$ principal minor. Unless $l=n$, the latter minor vanishes because either its first column (if $\nu_{l+1} < \tilde{\nu}_{l+1}$) or its first row (if $\nu_{l+1} > \tilde{\nu}_{l+1}$) contains only zeros.
This completes the proof of the asserted orthogonality and normalization of our eigenfunctions.
\end{proof}

\begin{note}  Macdonald's ninth variation of the Schur polynomials in  \cite{mac:schur} associates a generalized Schur polynomial
to any basis of monic univariate polynomials of degree $k=0,1,2,\ldots $ (cf.  also \cite{naketal:tableau,ser-ves:jacobi} and references therein).
From this perspective, our Schur polynomials $\textsc{s}^{(\nu)}_\lambda$ \eqref{schur} amount to the particular case of generalized Schur polynomials associated with the elliptic Racah polynomials from \cite{die-gor:racah}. Since the elliptic Racah polynomial $P_k (E)$ \eqref{E-Racah} is monic and of degree $k$, the determinant $\textsc{a}^{(\nu)}_0$ \eqref{schur} can be brought
to the Vandermonde form by unitriangular row-operations. The denominator of  $\textsc{s}^{(\nu)}_\lambda$ thus factorizes in terms of the roots of the top-degree elliptic Racah polynomial $P_{n+m}(E)$ as follows:
\begin{equation}
\textsc{a}^{(\nu)}_0= \prod_{1\leq j<k\leq n}(E_{n-j+\nu_j}-E_{n-k+\nu_k})  .
\end{equation}
The monotonicity of the roots $E_0,\ldots ,E_{n+m-1}$ \eqref{E-racah:roots} guarantees that the denominator in question does not vanish.
\end{note}

\section*{Acknowledgements}
The work of JFvD was supported in part by the {\em Fondo Nacional de Desarrollo
Cient\'{\i}fico y Tecnol\'ogico (FONDECYT)} Grant \# 1210015. TG was supported in part by the NKFIH Grant K134946.

\bigskip\noindent
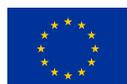
\parbox{.135\textwidth}{\begin{tikzpicture}[scale=.03]
\fill[fill={rgb,255:red,0;green,51;blue,153}] (-27,-18) rectangle (27,18);  
\pgfmathsetmacro\inr{tan(36)/cos(18)}
\foreach \i in {0,1,...,11} {
\begin{scope}[shift={(30*\i:12)}]
\fill[fill={rgb,255:red,255;green,204;blue,0}] (90:2)
\foreach \x in {0,1,...,4} { -- (90+72*\x:2) -- (126+72*\x:\inr) };
\end{scope}}
\end{tikzpicture}} \parbox{.85\textwidth}{This project has received funding from the European Union's Horizon 2020 research and innovation programme under the Marie Sk{\l}odowska-Curie grant agreement No 795471.}

\bibliographystyle{amsplain}

\end{document}